\RequirePackage{fix-cm}

 \def\bR{{\mathbf{R}}}
 \def\bZ{{\mathbf{Z}}}

\documentclass[smallextended]{svjour3}       
\smartqed  
\usepackage{graphicx}
\begin{document}
\title{Self-similar shapes under Errera division rule on the cone
}
%

\subtitle{Effect of meristem curvature on plant cell division}


\author{Etienne Couturier}

\institute{Laboratoire MSC, Universit\'e Paris-Diderot, Paris, France\\
              \email{etienne.couturier@univ-paris-diderot.fr}
}

\date{Received: date / Accepted: date}

\maketitle

\begin{abstract}
This study provides a general construction method of cell shape invariant by the Errera rule of division on a cone and provides analytical bounds for the apical angle of the cone on which these cells are connected and thus biologically meaningful. This idealized model highlights how the curvature of the tissue can influence Errera rule. 

\end{abstract}

\section{Introduction}
\label{intro}
The empirical division rule formulated by Errera in 1886 \cite{Errera1886} states that a given cell will divide in two daughters of equal area while minimizing the length of the added perimeter; nevertheless exceptions to the rule were numerous and the rule was discarded. When dividing a cell in two daughters, lot of solutions coexist among which only one is a global minimum of added perimeter whereas the other are only local minima; lots of the aformentioned exceptions have recently been explained by generalizing the Errera rule saying local minima are represented following a probability law proportional to the relative quality of the minimum \cite{Besson2011}; this modified rule has been proved to be valid only on the convex part of \textit{A. thaliana} meristem failing to describe division where stresses are too anisotropic such as at the junction between meristem and primordium \cite{Louveaux2016}.\\
The associated minimization problem, coined the "fencing problem" by mathematicians, was first studied and solved in the plane in 1914 by Wiener \cite{Wiener1914},\cite{Wang2015}: for closed piecewise $C^1$ planar contour, the shortest dividing line is a constant curvature arc (arc of circle) orthogonal to the contour at its anchorages. In general, plant meristem are surfaces with positive gaussian curvature but the existing studies are all plane. As cones are isometric to a plane deprived from an angular sector, fencing solutions on the cone are also constituted by arcs of circles; though the Errera rule on the cone will lead to richer behaviours than on the plane as not all the cuts will lead to identical fencing solutions once mapped back on the cone. Motivated by the morphology on cells in a vegetal tissue (\cite{Couturier}), this article is interested in shapes whose one of the daughters is similar to the mother when dividing following the Errera rule; the focus is on the influence on the apical angle of the cone on these shapes. \\
The philosophy of the article is similar to the approach of Hofmeister and Van Iterson who provided the deepest understanding of phyllotaxy through two classical analytical models idealizing primordia on the meristem by tangent smaller disks on a bigger disk \cite{Hofmeister1868} or tangent disks on a cylinder \cite{Iterson1907}; these models are still guiding simulations and experiments. Similarly in order to understand how the coupling between the meristem curvature and the Errera rule influences division, the whole curvature of the meristem is concentrated in a sole point idealizing the meristem by a cone.
\section{Definitions}

\subsection{Basic definitions, notations and abbreviations}
\label{defi}
\begin{enumerate}
\item Let O the point of coordinate $(0,0)$.
\item Let $\arccos$ stands for the $[0,\pi]$ determination of the arc cosinus.
\item The positive orientation of the angle is counterclockwise.
\item Let $\theta_1$, $\theta_2$ in $\bR$, note $\left \lfloor{ \phi}\right \rfloor_{]\theta_1,\theta_2]}$ its $]\theta_1,\theta_2]$ determination.
\item Let  $\overrightarrow{a}$, $\overrightarrow{b}$ be two planar vectors. $(\overrightarrow{a},\overrightarrow{b})$ stands for the oriented angle between  $\overrightarrow{a}$ and $\overrightarrow{b}$ determined between $-\pi$ and $\pi$.
\item Let  $\overrightarrow{a}$, $\overrightarrow{b}$ be two consecutive sides of a given polygon. $(\overrightarrow{a},\overrightarrow{b})_i$ stands for the angle between  $\overrightarrow{a}$ and $\overrightarrow{b}$ interior to the polygon determined between $0$ and $2\pi$.
\item \textbf{A cone}: Let $\gamma>0$. A cone with $\gamma$ for planar angular span is defined herein by a smooth immersion $S$ from $\bR^*_+\times\bR/\gamma\bR$ to $\bR^3$
whose metric reads:
\begin{eqnarray*}
\left(\frac{\partial S}{\partial \rho}\right)^2=1,\ 
\left(\frac{\partial S}{\partial \rho}\right)\cdot\left(\frac{\partial S}{\partial \phi}\right)=0,\ 
\left(\frac{\partial S}{\partial \phi}\right)^2=\rho^2,
\end{eqnarray*}
\begin{eqnarray}
\label{Def_metrique}
\end{eqnarray}
$(\rho,\phi)$ being the generic element of $\bR^*_+\times\bR/\gamma\bR$.\\
For an axisymmetric cone, $\gamma$ is necessarily inferior to $2\pi$ and $\Theta$, the apical angle of the cone is linked to $\gamma$ by the relationship:
$\gamma=2\pi\sin(\frac{\Theta}{2})$.\\
$\alpha=2\pi-\gamma$ is coined the angular defect of the cone for $\alpha>0$ and the angular excess for $\alpha<0$.
\item \textbf{Constant geodesic curvature line}:\\
 $\mathcal{C}$, a constant geodesic curvature line (abbreviation CCL) on the cone is defined by the triplet $(R_0,D_0,\phi_0)$ where $R,\ D$ are two positive real numbers, and $\phi_0$ is an angular number of $\bR/\gamma\bR$. $R_0$ is the radius of geodesic curvature of the line.\\
\textbf{Finite perimeter CCL}: If $R_0<D_0$, the radial coordinate of $\mathcal{C}$ is the union of:
\begin{eqnarray*}
\forall \phi\in[\phi_0-\Delta\phi_0,\phi_0+\Delta\phi_0],&&\\
&&\rho_+(\phi)=D_0\cos(\phi-\phi_0)+\sqrt{R_0^2-D_0^2\sin(\phi-\phi_0)^2}\\
&&\rho_-(\phi)=D_0\cos(\phi-\phi_0)-\sqrt{R_0^2-D_0^2\sin(\phi-\phi_0)^2}
\end{eqnarray*}
$\Delta\phi_0$ standing for $\arcsin(\frac{R_0}{D_0})$.\\
$\mathcal{C}$ is given by the union of two lines defined for $\phi$ in $[\phi_0-\Delta\phi_0,\phi_0+\Delta\phi_0]$ by $S(\rho_+(\phi),\phi)$ and $S(\rho_-(\phi),\phi)$.\\
$\rho_+$ corresponds to the part of the curve which is the further from the apex, and $\rho_-$ corresponds to the part of the curve which is the nearer from the apex. $\mathcal{C}$ is a circle of the cone.\\
\textbf{Non-finite perimeter CCL}:  If $R_0\geq D_0$, the radial coordinate of $\mathcal{C}$ is 
\begin{eqnarray*}
\forall \phi\in ]-\infty,\infty[,&&\\
&&\rho_+(\phi)=D_0\cos(\phi-\phi_0)+\sqrt{R_0^2-D_0^2\sin(\phi-\phi_0)^2}
\end{eqnarray*}\\
$\mathcal{C}$ is given by:
\begin{eqnarray*}
\forall \phi \in ]-\infty,\infty[,\ S(\rho_+(\phi),\phi)
\end{eqnarray*}
As $\rho_+$ is $2\pi$-periodic and $S$ is $\gamma$-periodic for the second coordinate, the length of $\mathcal{C}$ is non-finite  and $\mathcal{C}$ is in general a curve dense inside the annulus centered at the apex whose minimal radius is $R_0-D_0$ and maximal radius $R_0+D_0$ unless either $\gamma$ is a rational mutiple of $2\pi$ or $R_0$ equals $D_0$.\\\\
\item \textbf{Piecewise constant curvature closed contour orthogonal at the corners (CCAOC)}: \\
Let $N$ a positive integer. Let $\mathcal{P}$, a $N$-sided closed contour of the cone constituted by $N$ constant curvature arcs (abbreviated by CCAs), $(\mathcal{A}_{i})_{i\in\bZ / N \bZ}$ counterclockwise oriented around the cone apex and positively orthogonal at their intersection. Let $s_{\mathcal{P}}$ be a counterclockwise oriented natural parametrization of $\mathcal{P}$. $\mathcal{P}$ is unequivocally defined by:
\begin{enumerate}
\item The triplets $(R_i,D_i,\phi_i)_{i\in\bZ / N \bZ}$ specifying the CCL $(\mathcal{C}_i)_{i\in\bZ / N \bZ}$ on which lay the $(\mathcal{A}_{i})_{i\in\bZ / N \bZ}$.
\item The vertices $(P_i)_{i\in\bZ / N \bZ}$ of coordinates $(\rho_{P_i},\phi_{P_i})_{i\in\bZ / N \bZ}$. In this article, the convention is: 
$$P_i\in \mathcal{A}_{i}\bigcap\mathcal{A}_{i+1}.$$
\end{enumerate}
For a given $i$ in $\bZ / N \bZ$, the condition of orthogonality reads for $\mathcal{A}_{i}$ and $\mathcal{A}_{i+1}$ which are forming a $\frac{\pi}{2}$ oriented angle at their intersection(s):
\begin{eqnarray}
\left(\mathcal{A}_i^\prime(s^-_{P_i}), \mathcal{A}_{i+1}^\prime(s^+_{P_i})\right)=\frac{\pi}{2}\label{orth90}.
\end{eqnarray}
Here is a list of notation associated to the CCAOC $\mathcal{P}$ used in the text (The index $\mathcal{P}$ is omitted in the text when possible to alleviate the notations):
\begin{enumerate}
\item Let $l_{\mathcal{P}}$, the perimeter of $\mathcal{P}$.
\item Let $(\rho_{\mathcal{P}},\phi_{\mathcal{P}})$ be the closed contour in the coordinate space $\bR^*_+\times\bR/\gamma\bR$ corresponding to $\mathcal{P}$.
\item Let $(\delta\phi_j)_{j\in \bZ / N \bZ}$ defined by:
\begin{eqnarray}
\forall j\in\bZ / N \bZ, \delta\phi_j=\phi_{j+1}-\phi_j.\label{Def_deltaphi}
\end{eqnarray}
 \item Let $\sigma_{\mathcal{P}}$ the permutation associating the ascending order of the curvature radii of the $(\mathcal{A}_{i})_{i\in\bZ / N \bZ}$ to the index $i$ of $\bZ / N \bZ$.
\item Let $(o_i)_{i\in\bZ / N \bZ}$ in $\{-1,1\}^N$ be the orientation of $(\mathcal{A}_{i})_{i\in\bZ / N \bZ}$: For $i$ in $\bZ / N \bZ$, $o_i$ equals $-1$ for a clockwise rotation around the curvature center of $\mathcal{A}_{i}$ and $o_i$ equals $1$ for a counterclockwise rotation around the curvature center.
\item Let $\mathcal{P}^{*}$ the unclosed contour of the plane defined for $m$ in $\bZ$ and $s$ in $]ml_{\mathcal{P}},(m+1)l_{\mathcal{P}}]$ by:
\begin{eqnarray}
\mathcal{P}^{*}(s)=\rho_{\mathcal{P}}(s^*)(\cos(\phi_{\mathcal{P}}(s^*)+m\gamma),\sin(\phi_{\mathcal{P}}(s^*)+m\gamma))\label{P_plan}
\end{eqnarray}
with:
\begin{eqnarray*}
s^*=s-ml_{\mathcal{P}}.
\end{eqnarray*}
If $\gamma$ and $2\pi$ are not rational multiple, $\mathcal{P}^{*}$ is generically an unclosed infinite contour of the plane. $\mathcal{P}$ and $\mathcal{P}^{*}([0,\gamma])$ are isometric.\\
We also note for $j$ in $\{0,\cdots,N-1\}$ and $m$ in $\bZ$, $\mathcal{C}^*_{j+mN}$ (resp. $\mathcal{A}^*_{j+mN}$) the CCL and CCA of $\mathcal{P}^*$ image of $\mathcal{C}_{j}$ (resp. $\mathcal{A}_{j}$) rotated with an angle $m\gamma$ around $O$. For example,
$\mathcal{A}^*_{j+mN}$ is defined for $s$ in $[s_{P_{j}}+ml_{\mathcal{P}},s_{P_{j+1}}+ml_{\mathcal{P}}]$ by:
\begin{eqnarray}
\mathcal{A}^*_{j+mN}(s)=\rho_{\mathcal{A}_{j}}(s^*)(\cos(\phi_{\mathcal{P}}(s^*)+m\gamma),\sin(\phi_{\mathcal{P}}(s^*)+m\gamma)).\label{A_plan}
\end{eqnarray}
with:
\begin{eqnarray*}
s^*=s-ml_{\mathcal{P}}.
\end{eqnarray*}
Let define:
$$\delta\phi^*_{j+mN}=\left \lfloor{\delta\phi_j}\right \rfloor_{[-\pi,\pi]}.$$
\end{enumerate}
\item  \textbf{Halving solution (HS)}:
A CCA halving the area of a CCAOC while being either orthogonal to the contour of the CCAOC, or not touching the contour of CCAOC is coined an halving solution, abbreviated by HS.

\end{enumerate}
\section{Self-similar solution under Errera rule of division}
The theorem is proved using a technical Lemma in annex deriving some properties of CCAOCs whose the distances from the curvature centers of each CCA to the cone apex and their radii of curvature share a common ratio, $\rho$. The theorem both proves the CCAs of a self-similar $N$-sided CCAOC ($N$ being an integer) share a constant $\rho$ and the radii of curvature of each CCA are distributed around the CCAOC abiding to an invertible arithmetic sequence in $\bZ/N\bZ$. Results of the lemma are then used to determine the geometric properties of the self-similar CCAOC induced by such a permutation when two sides become tangent.

\begin{theorem}
Let $N>0$ an integer. Let $\mathcal{P}$ and $\mathcal{P}_{\frac{1}{2}}$ two $N$-sided CCAOCs on a cone of planar angular span $\gamma$. They are constituted by the arcs $(\mathcal{A}_{j})_{j\in\bZ / N \bZ}$ (\textit{resp.}  $(\mathcal{A}_{j,\frac{1}{2}})_{j\in\bZ / N \bZ}$) laying on the CCLs, $(\mathcal{C}_{j})_{j\in\bZ / N \bZ}$, (\textit{resp.}  $(\mathcal{C}_{j,\frac{1}{2}})_{j\in\bZ / N \bZ}$) whose parameters are $(R_j,D_j,\phi_j)_{j\in \bZ / N \bZ}$ (\textit{resp.}  $(R_{j,\frac{1}{2}},D_{j,\frac{1}{2}},\phi_{j,\frac{1}{2}})_{j\in \bZ / N \bZ}$).\\
Let suppose $\mathcal{P}_{\frac{1}{2}}$ is the daughter of $\mathcal{P}$ containing the apex of the cone and corresponding to an halving solution (HS). Let also suppose $\mathcal{P}_{\frac{1}{2}}$ can be mapped on $\mathcal{P}$ by $h$ a combination of an homothety and a rotation centered at the apex.\\
\begin{enumerate}
\item Then:
\begin{eqnarray}
\forall j\in\bZ / N \bZ,R_{\sigma(j)}&=&(\sqrt{2})^j R_{\sigma(0)},\label{eqR}\\
 D_{\sigma(j)}&=&(\sqrt{2})^j D_{\sigma(0)}\label{eqD}
\end{eqnarray}
$\sigma$ being an arithmetic progression of $\bZ / N \bZ$ whose common difference $k$ is prime with $N$. The maximum number of self-similar shape with orientation is thus the Euler totient function $\phi(N)$.
\item $(\delta\phi_j)_{j\in \bZ/N\bZ}$ takes $k^\prime$ times the value $\delta\phi_{\sigma(-1)}$ and  $N-k^\prime$ times the value $\delta\phi_{\sigma(0)}$ where $k^\prime$ is an inverse of $k$ in $\bZ / N \bZ$. More precisely:
\begin{itemize}
\item for $j$ in  $\sigma\left(\{0,\cdots,N-k^\prime-1\}\right)$:
\begin{eqnarray}
\delta\phi_{j}=\delta\phi_{\sigma(0)}\label{valdphi1};
\end{eqnarray}
\item for $j$ in $\sigma\left(\{N-k^\prime,\cdots,N-1\}\right)$:
\begin{eqnarray}
\delta\phi_{j}=\delta\phi_{\sigma(-1)}.\label{valdphi2}
\end{eqnarray}
\end{itemize}
\item Let $$\rho=\frac{R_0^2}{D_0^2}.$$
The $(\delta\phi^*_j)_{j\in \bZ/N\bZ}$ are given by the formula:
\begin{itemize}
\item for $j$ in $\sigma(\{0,\cdots,N-k^\prime-1\})$:
\begin{eqnarray}
\delta\phi^*_j=\epsilon_{\sigma(0)}\arccos\left((1-\rho)\frac{(\sqrt{2})^{k^\prime}+(\sqrt{2})^{-k^\prime}}{2}\right)\label{expphi1100};
\end{eqnarray} 
\item for $j$ in $\sigma(\{N-k^\prime,\cdots,N-1\})$
\begin{eqnarray}
\delta\phi^*_j=\epsilon_{\sigma(-1)}\arccos\left((1-\rho)\frac{(\sqrt{2})^{N-k^\prime}+(\sqrt{2})^{k^\prime-N}}{2}\right).\label{expphi1200}
\end{eqnarray} 
\end{itemize}
\item Let $i\in\bZ / N \bZ$. Let suppose there exists $i$ in $\bZ / N \bZ$ such as $\mathcal{A}_{i}$ is tangent to $\mathcal{A}_{i+1}$. It implies $\rho$ belongs to 
$$\{0,1-\frac{\varepsilon_{i,i+2}}{\tau_{i,i+1}\tau_{i+1,i+2}},1\},$$
with $\varepsilon_{i,i+2}=1$ in case of external tangency or $\varepsilon_{i,i+2}=-1$ in case of internal tangency.
\begin{itemize}
\item In the case,
\begin{eqnarray*}
\rho=0,
\end{eqnarray*}
$\gamma$ is not a real number.
\item In the case,
\begin{eqnarray*}
\rho=1-\frac{\varepsilon_{i,i+2}}{\tau_{i,i+1}\tau_{i+1,i+2}}, 
\end{eqnarray*}
there exists a positive integer $m$ such as
if neither $i$, nor $i+1$ belong to $\sigma(\{N-k^\prime,\cdots,N-1\})$:
\begin{eqnarray*}
\gamma=&&(N-k^\prime)\left(\pi(1-\epsilon_{\sigma(0)})+\epsilon_{\sigma(0)}\arccos\left(\frac{2\varepsilon_{i,i+2}}{(\sqrt{2})^{-k^\prime}+(\sqrt{2})^{k^\prime}}\right)\right)\\&&+k^\prime\left(\pi(1-\epsilon_{\sigma(-1)})+\epsilon_{\sigma(-1)} \arccos\left(\frac{2\varepsilon_{i,i+2}((\sqrt{2})^{N-k^\prime}+(\sqrt{2})^{k^\prime-N})}{((\sqrt{2})^{k^\prime}+(\sqrt{2})^{-k^\prime})^2}\right)\right)+2m\pi
\end{eqnarray*}
else:
\begin{eqnarray*}
\gamma=(N-k^\prime)\left(\pi(1-\epsilon_{\sigma(0)})+\epsilon_{\sigma(0)} \arccos\left(\frac{2\varepsilon_{i,i+2}}{(\sqrt{2})^{N-k^\prime}+(\sqrt{2})^{k^\prime-N}}\right)\right)\\
+k^\prime\left(\pi(1-\epsilon_{\sigma(-1)})+\epsilon_{\sigma(-1)}\arccos\left(\frac{2\varepsilon_{i,i+2}}{(\sqrt{2})^{-k^\prime}+(\sqrt{2})^{k^\prime}}\right)\right)+2m\pi
\end{eqnarray*}
with:
\begin{eqnarray}
\epsilon_i=sign\left(\varepsilon_{i,i+2}\left(R_{i+1}^2-R_i^2\right)\left(R_{i+1}^2-R_{i+2}^2\right)\right)\epsilon_{i+1}.
\end{eqnarray}
\item In the case,
\begin{eqnarray*}
\rho=1, 
\end{eqnarray*}
there exists a positive integer $m$ such as:
\begin{eqnarray*}
\gamma=(N-k^\prime)\left(\pi(1-\epsilon_{\sigma(0)})+\epsilon_{\sigma(0)}\frac{\pi}{2}\right)+k^\prime\left(\pi(1-\epsilon_{\sigma(-1)})+\epsilon_{\sigma(-1)} \frac{\pi}{2}\right)+2m\pi.
\end{eqnarray*}
\end{itemize}
\end{enumerate}
\end{theorem}
\begin{proof}
\begin{enumerate}
\item $\mathcal{P}$ and $\mathcal{P}_{\frac{1}{2}}$ are self-similar thus there exists $h$ a combination between an homothety and a rotation such as:
\begin{eqnarray*}
h(\mathcal{P})=\mathcal{P}_{\frac{1}{2}}.
\end{eqnarray*}
Let index the CCLs of $\mathcal{P}_{\frac{1}{2}}$ in the counterclockwise order such as  $\mathcal{C}_{0,\frac{1}{2}}$ the first CCL of the indexing is the image by $h$ of $\mathcal{C}_{0}$  the first CCL of the indexing of $\mathcal{P}$.\\
As the area of the mother is twice the area of the daughter, the ratio of the homothety-rotation $h$ is necessary $\sqrt{2}$. 
Because of the homothety ratio $\sqrt{2}$ between $\mathcal{P}$ and $\mathcal{P}_{\frac{1}{2}}$, the following relations are verified between the parameter of $\mathcal{P}$ and $\mathcal{P}_{\frac{1}{2}}$:
\begin{eqnarray}
\forall j\in \{0,\cdots,n-1\},R_{\sigma(j),\frac{1}{2}}&=&\frac{R_{\sigma(j)}}{\sqrt{2}};\\
 D_{\sigma(j),\frac{1}{2}}&=&\frac{D_{\sigma(j)}}{\sqrt{2}}.\label{opoppp}
\end{eqnarray}
The division conserves the number of CCAs, thus only one CCA of $\mathcal{P}$ disappears and only one new CCA of $\mathcal{P}_{\frac{1}{2}}$ is created. The CCA which disappears is the one of highest radius of curvature $R_{\sigma(N-1)}$ and the one which is appearing is the one of lowest radius of curvature $\frac{R_{\sigma(0)}}{\sqrt{2}}$; all the other CCAs are wholly conserved during the division except the anchorage CCAs which are only partly conserved. The CCL of lower radius of curvature of $\mathcal{P}_{\frac{1}{2}}$ is the CCL corresponding to the halving solution $\mathcal{C}_{HS}$ while for $j$ in $\{1,\cdots,N-1\}$, the $j^{th}$ CCL of $\mathcal{P}_{\frac{1}{2}}$ for the radius of curvature corresponds to the $(j-1)^{th}$ CCL of $\mathcal{P}$ for the radius of curvature which reads:
\begin{eqnarray}
\mathcal{C}_{\sigma(0),\frac{1}{2}}&=&\mathcal{C}_{HS};\label{eqcote0}\\
\forall j\in \bZ / N \bZ\setminus\{0\},\mathcal{C}_{\sigma(j),\frac{1}{2}}&=&\mathcal{C}_{\sigma(j-1)}\label{eqcotej}
\end{eqnarray}
which implies the following equalities between the radii of curvature: 
\begin{eqnarray}
\forall j\in \bZ / N \bZ,R_{\sigma(j),\frac{1}{2}}&=&R_{\sigma(j-1)};\\
				      D_{\sigma(j),\frac{1}{2}}&=&D_{\sigma(j-1)}.\label{opopppo}
\end{eqnarray}
(\ref{opoppp}) combined with (\ref{opopppo}) gives:
\begin{eqnarray}
\forall j\in \bZ / N \bZ,\frac{R_{\sigma(j)}}{\sqrt{2}}&=&R_{\sigma (j-1)},\\
\frac{D_{\sigma(j)}}{\sqrt{2}}&=&D_{\sigma (j-1)}.
\end{eqnarray}
The (\ref{eqR},\ref{eqD}) formula are easily obtained by induction on $j$:
\begin{eqnarray}
\forall j\in\bZ / N \bZ,R_{\sigma(j)}&=&(\sqrt{2})^j R_{\sigma(0)};\\
				      D_{\sigma(j)}&=&(\sqrt{2})^j D_{\sigma(0)}.
\end{eqnarray}
The permutation $\sigma$ of the mother induces a permutation $\sigma_{\frac{1}{2}}$ for the daughter associating the ascending order of the radii of curvature of the daughter to some counterclockwise indexing of its sides: for $j$ in $\{0,\cdots,N-2\}$, the counterclockwise index of the $\mathcal{P}$ CCL whose radius of curvature is the $(j)^{th}$ in ascending order is inherited by the $\mathcal{P}_{\frac{1}{2}}$ CCL whose radius of curvature is the $(j+1)^{th}$ and the counterclockwise index of the $\mathcal{P}$ CCL whose radius of curvature is the highest (the $(N-1)^{th}$) is inherited by the newly created halving solution which corresponds to the $\mathcal{P}_{\frac{1}{2}}$ CCL whose curvature radius is the lowest (the $0^{th}$). This definition can be formally written:
\begin{eqnarray}
\forall j\in \bZ / N \bZ,\sigma_{\frac{1}{2}}(j)=\sigma(j-1).\label{oupipoliu}
\end{eqnarray}
The trigonometric order induced by the homothety-rotation $h$ on $\mathcal{P}_{\frac{1}{2}}$ is equal up to a circular permutation of offset $k$ (corresponding to the possible difference in the origin of the counterclockwise indexing) which implies:
\begin{eqnarray}
\sigma=\sigma_{\frac{1}{2}}+k.\label{oupipolio}
\end{eqnarray}
Combining (\ref{oupipoliu}) and (\ref{oupipolio}) yields:
\begin{eqnarray}
\forall j\in \bZ / N \bZ,\sigma(j)=\sigma(j-1)+k.\label{diffarithmetique}
\end{eqnarray}
(\ref{diffarithmetique}) implies $\sigma$ is an arithmetic progression of common difference $k$ which reads:
\begin{eqnarray}
\forall j\in \bZ / N \bZ,\sigma(j)=\sigma(0)+jk.\label{arithmetique}
\end{eqnarray}
In particular, (\ref{arithmetique}) implies:
\begin{eqnarray}
\sigma(\sigma^{-1}(\sigma(0)+1))&=& \sigma(0)+(\sigma^{-1}(1+\sigma(0)))k.\label{parti} 
\end{eqnarray}
The left side of (\ref{parti}) can be rewritten:
\begin{eqnarray}
\sigma(\sigma^{-1}(\sigma(0)+1))&=&\sigma(0)+1\label{left}
\end{eqnarray}
substituting (\ref{left}) on the left of (\ref{parti}) gives:
\begin{eqnarray}
\sigma(0)+1&=& \sigma(0)+(\sigma^{-1}(1+\sigma(0)))k.\label{eqooppp}
\end{eqnarray}
Canceling $\sigma(0)$ on both sides of (\ref{eqooppp}) yields:
\begin{eqnarray*}
1&=&(\sigma^{-1}(\sigma(0)+1))k
\end{eqnarray*}
which implies $k$ is invertible in $\bZ / N \bZ$ and thus prime with $N$.\\\\

\item Let $j\in \bZ/N\bZ\setminus\{\sigma(0),\sigma(0)-1\}$. By choice of the indexing of $\mathcal{P}_{\frac{1}{2}}$, for any $j$ in $\bZ / N \bZ$
\begin{eqnarray*}
h(\mathcal{C}_{j})&=&\mathcal{C}_{j,\frac{1}{2}},\\ h(\mathcal{C}_{j+1})&=&\mathcal{C}_{j+1,\frac{1}{2}}.
\end{eqnarray*}
\begin{eqnarray}
\label{eqhj}
\end{eqnarray}
As neither $j=\sigma(0)$ nor  $j+1=\sigma(0)$, (\ref{eqcotej}) implies:
\begin{eqnarray*}
\mathcal{C}_{j,\frac{1}{2}}&=&\mathcal{C}_{\sigma(\sigma^{-1}(j)-1)},\\
\mathcal{C}_{j+1,\frac{1}{2}}&=&\mathcal{C}_{\sigma(\sigma^{-1}(j+1)-1)}.
\end{eqnarray*}
Using (\ref{diffarithmetique}), it rewrites: 
\begin{eqnarray*}
\mathcal{C}_{j,\frac{1}{2}}&=&\mathcal{C}_{\sigma(\sigma^{-1}(j))-k},\\
\mathcal{C}_{j+1,\frac{1}{2}}&=&\mathcal{C}_{\sigma(\sigma^{-1}(j+1))-k},
\end{eqnarray*}
which simplifies into:
\begin{eqnarray*}
\mathcal{C}_{j,\frac{1}{2}}&=&\mathcal{C}_{j-k},\\
\mathcal{C}_{j+1,\frac{1}{2}}&=&\mathcal{C}_{j+1-k}.
\end{eqnarray*}
Combining with (\ref{eqhj}), it gives:
\begin{eqnarray*}
h(\mathcal{C}_{j})&=&\mathcal{C}_{j-k},\\
h(\mathcal{C}_{j+1})&=&\mathcal{C}_{j+1-k}.
\end{eqnarray*}
As $h$ is an isometric transformation, the angle are conserved thus implying: 
\begin{eqnarray}
\delta\phi_{j}=\delta\phi_{j-k}\label{equallol}.
\end{eqnarray}
Let $j$ in $$\sigma(0) +k\{1,\cdots,N-k^\prime-1\},$$ $j+1$ can be decomposed in
\begin{eqnarray}
j+1=j+kk^\prime,
\end{eqnarray}
thus $j+1$ belongs to
\begin{eqnarray}
\sigma(0) +k\{1+k^\prime,\cdots,N-1\}
\end{eqnarray}
and is always different from $\sigma(0)$. (\ref{equallol}) can iteratively be applied, implying by transitivity of the equality relationship for $j$ in 
$$\sigma(\{0,\cdots,N-k^\prime-1\})=\sigma(0) +k\{0,\cdots,N-k^\prime-1\},$$
\begin{eqnarray}
\delta\phi_{j}=\delta\phi_{\sigma(0)}\label{equa0mpp}.
\end{eqnarray}
Let $j$ in $$\sigma(0) +k\{N-k^\prime+1,\cdots,N-1\},$$ $j+1$ can be decomposed in
\begin{eqnarray}
j+1=j+kk^\prime,
\end{eqnarray}
thus $j+1$ belongs to
\begin{eqnarray}
\sigma(0) +k\{1,\cdots,k^\prime-1\}
\end{eqnarray}
and is always different from $\sigma(0)$. (\ref{equallol}) can iteratively be applied, implying by transitivity of the equality relationship for $j$ in 
$$\sigma(\{N-k^\prime,\cdots,N-1\})=\sigma(0) +k\{N-k^\prime,\cdots,N-1\},$$
\begin{eqnarray}
\delta\phi_{j}=\delta\phi_{\sigma(-1)}\label{equa0mpp2}.
\end{eqnarray}
Finally $(\delta\phi_j)_{j\in \bZ/N\bZ}$ takes $k^\prime$ times the value $\delta\phi_{\sigma(-1)}$ and  $N-k^\prime$ times the value $\delta\phi_{\sigma(0)}$ which finishes to prove (\ref{valdphi1},\ref{valdphi2}).

\item For any $j$ in $\bZ / N \bZ$, the ratio between (\ref{eqR},\ref{eqD}) formula read:
\begin{eqnarray*}
\frac{R_{\sigma(j)}}{D_{\sigma(j)}}&=&\frac{(\sqrt{2})^j R_{\sigma(0)}}{(\sqrt{2})^j D_{\sigma(0)}}
\end{eqnarray*}
which simplifies into:
\begin{eqnarray*}
\frac{R_{\sigma(j)}}{D_{\sigma(j)}}&=&\frac{R_{\sigma(0)}}{D_{\sigma(0)}}.
\end{eqnarray*}
As $\sigma$ is a permutation of $\bZ / N \bZ$, for any $j$:
\begin{eqnarray*}
\frac{R_{j}}{D_{j}}=\sqrt{\rho}.
\end{eqnarray*}
The hypothesis of Lemma \ref{lm:casbase2}.1 are verified thus for any $j$ in $\bZ / N \bZ$,
\begin{eqnarray}
\cos(\delta\phi_j)=(1-\rho)\tau_{j,j+1}
\end{eqnarray}
which can be rewritten using the notation (\ref{okolo}):
 \begin{eqnarray}
\cos(\delta\phi_{j})&=&(1-\rho)\frac{\kappa_j+\frac{1}{\kappa_j}}{2}\label{coskappa}
\end{eqnarray}
substituting (\ref{opoppp}) in (\ref{okolo}), we obtain:
\begin{eqnarray*}
\kappa_j&=&\frac{R_0(\sqrt{2})^{\sigma^{-1}(j)}}{R_0(\sqrt{2})^{\sigma^{-1}(j+1)}}
\end{eqnarray*}
which we rewrite:
\begin{eqnarray}
\kappa_j&=&(\sqrt{2})^{Inj(\sigma^{-1}(j))-_{\bZ}Inj(\sigma^{-1}(j+1))}\label{dlout}
\end{eqnarray}
$Inj$, being the canonical injection from $\bZ /N \bZ$ to $\bZ$, $-_{\bZ}$ is the difference in $\bZ$.\\
As $\sigma^{-1}$ reads:
\begin{eqnarray*}
\forall j \in \bZ /N \bZ,\ \sigma^{-1}(j)=k^\prime j-k^\prime\sigma(0),
\end{eqnarray*}
if $j$ lays in $\sigma(\{N-k^\prime,\cdots,N-1\}),$
$\sigma^{-1}(j)$ lays in $\{N-k^\prime,\cdots,N-1\}$, it implies:
$$Inj(\sigma^{-1}(j)+k^\prime)=Inj(\sigma^{-1}(j))+_{\bZ}k^\prime-N,$$
thus (\ref{dlout}) implies:
\begin{eqnarray}
\kappa_j&=&(\sqrt{2})^{N-k^\prime}\label{e_cas1}
\end{eqnarray}
else if $j$ does not lay in $\sigma(\{N-k^\prime,\cdots,N-1\}),$ it implies:
$$Inj(\sigma^{-1}(j)+k^\prime)=Inj(\sigma^{-1}(j))+_{\bZ}k^\prime$$
and (\ref{dlout}) implies:
\begin{eqnarray}
\kappa_j&=&(\sqrt{2})^{k^\prime}.\label{e_cas2}
\end{eqnarray}
Substituting in (\ref{coskappa}) with (\ref{e_cas1}) and (\ref{e_cas2}) gives:
\begin{itemize}
\item For $j$ in $\sigma(\{0,\cdots,N-k^\prime-1\})$, 
\begin{eqnarray*}
\cos(\delta\phi_j)=\rho\frac{(\sqrt{2})^{k^\prime}+(\sqrt{2})^{-k^\prime}}{2}
\end{eqnarray*}
\item For $j$ in $\sigma(\{N-k^\prime,\cdots,N-1\})$,
\begin{eqnarray*}
\cos(\delta\phi_j)=\rho\frac{(\sqrt{2})^{N-k^\prime}+(\sqrt{2})^{k^\prime-N}}{2}.
\end{eqnarray*}
\end{itemize}
It implies for $j$ in $\sigma(\{0,\cdots,N-k^\prime-1\})$ there exists $\epsilon_j$ in $\{-1,1\}$ such as:
\begin{eqnarray}
\delta\phi^*_j=\epsilon_j\arccos\left(\rho\frac{(\sqrt{2})^{k^\prime}+(\sqrt{2})^{-k^\prime}}{2}\right),\label{expphi11}
\end{eqnarray}
and for $j$ in $\sigma(\{N-k^\prime,\cdots,N-1\})$, there exists $\epsilon_j$ in $\{-1,1\}$ such as:
\begin{eqnarray}
\delta\phi^*_j=\epsilon_j\arccos\left(\rho\frac{(\sqrt{2})^{N-k^\prime}+(\sqrt{2})^{k^\prime-N}}{2}\right).\label{expphi12}
\end{eqnarray}
(\ref{valdphi1},\ref{valdphi2}) implies $(\epsilon_j)_{j\in\{0,\cdots,N-1\}}$ takes only two values thus (\ref{expphi11}) rewrites for $j$ in $\sigma(\{0,\cdots,N-k^\prime-1\})$:
\begin{eqnarray}
\delta\phi^*_j=\epsilon_{\sigma(0)}\arccos\left(\rho\frac{(\sqrt{2})^{k^\prime}+(\sqrt{2})^{-k^\prime}}{2}\right),\label{expphi110}
\end{eqnarray} 
and (\ref{expphi12}) rewrites for $j$ in $\sigma(\{N-k^\prime,\cdots,N-1\})$,
\begin{eqnarray}
\delta\phi^*_j=\epsilon_{\sigma(N-1)}\arccos\left(\rho\frac{(\sqrt{2})^{N-k^\prime}+(\sqrt{2})^{k^\prime-N}}{2}\right),\label{expphi120}
\end{eqnarray} 
thus proving (\ref{expphi1100},\ref{expphi1200}).
 
\item As $\mathcal{A}_{i}$ is tangent to $\mathcal{A}_{i+1}$, Lemma $1.2$ implies $\rho$ can only take three values:
\begin{enumerate}
\item If $\rho=0$, 
 \begin{eqnarray*}
 \cos(\delta\phi_i)=\frac{\kappa_i+\frac{1}{\kappa_i}}{2}
 \end{eqnarray*}
As by convexity, the arithmetic mean is superior to the geometric mean:
\begin{eqnarray*}
\cos(\delta\phi_i)\geq \sqrt{\kappa_i\frac{1}{\kappa_i}}=1
\end{eqnarray*}
As $R_i\neq R_{i+1}$, $\kappa_i\neq 1$ and the inequality is strict:
\begin{eqnarray*}
\cos(\delta\phi_i)>1.
\end{eqnarray*}
$\delta\phi_i$ is an imaginary number as well as the planar angular span $\gamma$.
\item If $\rho=1-\frac{\varepsilon_{i,i+2}}{\tau_{i,i+1}\tau_{i+1,i+2}}$, substituting (\ref{val_rho_int}) into (\ref{expphi1100},\ref{expphi1200}) gives:
 \begin{eqnarray*}
\cos(\delta\phi_{j})&=&\varepsilon_{i,i+2}\frac{2}{\kappa_i+\frac{1}{\kappa_i}}\frac{2}{\kappa_{i+1}+\frac{1}{\kappa_{i+1}}}\frac{\kappa_j+\frac{1}{\kappa_j}}{2}
\end{eqnarray*}
finally simplifying into:
\begin{eqnarray}
\cos(\delta\phi_{j})&=&\varepsilon_{i,i+2}\frac{2(\kappa_j+\frac{1}{\kappa_j})}{(\kappa_i+\frac{1}{\kappa_i})(\kappa_{i+1}+\frac{1}{\kappa_{i+1}})}.\label{couloucoulou}
\end{eqnarray}
(\ref{couloucoulou}) can now be further simplified by enumerating the different case for $i$, $i+1$ and $j$:
\begin{itemize}
\item If neither $i$, nor $i+1$, nor $j$ belong to  $\sigma(\{N-k^\prime,\cdots,N-1\})$, (\ref{e_cas1}) and (\ref{e_cas2}) imply $\kappa_i$, $\kappa_{i+1}$ and $\kappa_j$ are equal and (\ref{couloucoulou}) rewrites:
\begin{eqnarray}
\delta\phi^*_j=\epsilon_{\sigma(0)}\arccos\left(\frac{2\varepsilon_{i,i+2}}{\kappa_{i}+\frac{1}{\kappa_{i}}}\right);\label{eqoqu1}
\end{eqnarray}
substituting (\ref{e_cas2}) gives:
\begin{eqnarray}
\delta\phi^*_j=\epsilon_{\sigma(0)}\arccos\left(\frac{2\varepsilon_{i,i+2}}{(\sqrt{2})^{-k^\prime}+(\sqrt{2})^{k^\prime}}\right).\label{not_i_ip_j}
\end{eqnarray}
\item If neither $i$, nor $i+1$ belong to $\sigma(\{N-k^\prime,\cdots,N-1\})$ but only $j$, substituting (\ref{e_cas1}) and (\ref{e_cas2}) into (\ref{couloucoulou}) gives:
\begin{eqnarray}
\delta\phi^*_j=\epsilon_{\sigma(-1)}\arccos\left(\frac{2\varepsilon_{i,i+2}((\sqrt{2})^{N-k^\prime}+(\sqrt{2})^{k^\prime-N})}{((\sqrt{2})^{k^\prime}+(\sqrt{2})^{-k^\prime})^2}\right).\label{not_i_ip_but_j}
\end{eqnarray}
\item $\sigma^{-1}(i+1)=\sigma^{-1}(i)+k^\prime$ thus $i$ and $i+1$ can't simultaneously belong to $\sigma(\{N-k^\prime,\cdots,N-1\})$.\\
\item If $i$ and $j$ belong to $\sigma(\{N-k^\prime,\cdots,N-1\})$ but not $i+1$, (\ref{e_cas1}) and (\ref{e_cas2}) imply $\kappa_i$ and $\kappa_j$ are equal, finally (\ref{couloucoulou}) simplifies into:
\begin{eqnarray*}
\delta\phi^*_j=\epsilon_{\sigma(-1)}\arccos\left(\frac{2\varepsilon_{i,i+2}}{\kappa_{i+1}+\frac{1}{\kappa_{i+1}}}\right);
\end{eqnarray*}
substituting (\ref{e_cas2}) gives:
\begin{eqnarray}
\delta\phi^*_j=\epsilon_{\sigma(-1)}\arccos\left(\frac{2\varepsilon_{i,i+2}}{(\sqrt{2})^{-k^\prime}+(\sqrt{2})^{k^\prime}}\right).\label{i_notip_j}
\end{eqnarray}
A similar result would be obtain if $i+1$ and not $i$ belongs to $\sigma(\{N-k^\prime,\cdots,N-1\})$.
\item If $i$ belong to $\sigma(\{N-k^\prime,\cdots,N-1\})$ but neither $i+1$ neither $j$, (\ref{e_cas1}) and (\ref{e_cas2}) imply $\kappa_{i+1}$ and $\kappa_j$ are equal, finally (\ref{couloucoulou}) simplifies into:
\begin{eqnarray*}
\delta\phi^*_j=\epsilon_{\sigma(0)}\arccos\left(\frac{2\varepsilon_{i,i+2}}{\kappa_{i}+\frac{1}{\kappa_{i}}}\right);
\end{eqnarray*}
substituting (\ref{e_cas1}) gives:
\begin{eqnarray}
\delta\phi^*_j=\epsilon_{\sigma(0)}\arccos\left(\frac{2\varepsilon_{i,i+2}}{(\sqrt{2})^{N-k^\prime}+(\sqrt{2})^{k^\prime-N}}\right).\label{i_notip_notj}
\end{eqnarray}
A similar result would be obtain if $i+1$ and not $i$ belongs to $\sigma(\{N-k^\prime,\cdots,N-1\})$.
\end{itemize}
Lemma $1.3$ tells there exists a positive integer $m$ such as:
\begin{eqnarray*}
\gamma=\sum_{j\in\{0,\cdots,N-1\}}\big(\pi(1-\epsilon_j)+\delta\phi^*_j\big)+2m\pi.
\end{eqnarray*}
The sum can be divided in two subsums:
\begin{eqnarray*}
\sum_{j\in\{0,\cdots,N-k^\prime-1\}}\big(\pi(1-\epsilon_j)+\delta\phi^*_j\big)+\sum_{j\in\{N-k^\prime,\cdots,N-1\}}\big(\pi(1-\epsilon_j)+\delta\phi^*_j\big).
\end{eqnarray*}
\begin{eqnarray}
\label{sum_d}
\end{eqnarray}
(\ref{valdphi1},\ref{valdphi2}) imply (\ref{sum_d}) further simplifies, leading to:
\begin{eqnarray*}
\gamma=(N-k^\prime)\big(\pi(1-\epsilon_{\sigma(0)})+\delta\phi^*_{\sigma(0)}\big)+k^\prime\big(\pi(1-\epsilon_{\sigma(-1)})+\delta\phi^*_{\sigma(-1)}\big)+2m\pi.
\end{eqnarray*}
\begin{eqnarray}
\label{sum_s}
\end{eqnarray}
Substituting with (\ref{not_i_ip_j},\ref{not_i_ip_but_j}), in case neither $i$, nor $i+1$ belong to $\sigma(\{N-k^\prime,\cdots,N-1\})$ yields:
\begin{eqnarray*}
\gamma=&&(N-k^\prime)\left(\pi(1-\epsilon_{\sigma(0)})+\epsilon_{\sigma(0)}\arccos\left(\frac{2\varepsilon_{i,i+2}}{(\sqrt{2})^{-k^\prime}+(\sqrt{2})^{k^\prime}}\right)\right)\\
&+&k^\prime\left(\pi(1-\epsilon_{\sigma(-1)})+\epsilon_{\sigma(-1)}\arccos\left(\frac{2\varepsilon_{i,i+2}((\sqrt{2})^{N-k^\prime}+(\sqrt{2})^{k^\prime-N})}{((\sqrt{2})^{k^\prime}+(\sqrt{2})^{-k^\prime})^2}\right)\right)+2m\pi.
\end{eqnarray*}
\begin{eqnarray}
\label{Final1}
\end{eqnarray}
Substituting with (\ref{i_notip_j},\ref{i_notip_notj}), in case either $i$, or $i+1$ belong to $\sigma(\{N-k^\prime,\cdots,N-1\})$ yields:
\begin{eqnarray*}
\gamma=&&(N-k^\prime) \left(\pi(1-\epsilon_{\sigma(0)})+\epsilon_{\sigma(0)}\arccos\left(\frac{2\varepsilon_{i,i+2}}{(\sqrt{2})^{N-k^\prime}+(\sqrt{2})^{k^\prime-N}}\right)\right)\\
&+&k^\prime\left(\pi(1-\epsilon_{\sigma(-1)})+\epsilon_{\sigma(-1)}\arccos\left(\frac{2\varepsilon_{i,i+2}}{(\sqrt{2})^{-k^\prime}+(\sqrt{2})^{k^\prime}}\right)\right)+2m\pi.
\end{eqnarray*}
\begin{eqnarray}
\label{Final2}
\end{eqnarray}
\item If $\rho=1$, Lemma 1.3 tells there exists a positive integer $m$ such as substituting with (\ref{rootdphi01},\ref{rootdphi02}) in (\ref{expphi1100},\ref{expphi1200}) gives:
\begin{eqnarray*}
\gamma=(N-k^\prime)\left(\pi(1-\epsilon_{\sigma(0)})+\epsilon_{\sigma(0)}\frac{\pi}{2}\right)+k^\prime\left(\pi(1-\epsilon_{\sigma(-1)})+\epsilon_{\sigma(-1)} \frac{\pi}{2}\right)+2m\pi.
\end{eqnarray*}
\end{enumerate}
\end{enumerate}
\end{proof} 
\section{Conclusion}
In this article, self-similar shapes for the Errera division rule have been exhaustively constructed on the cone; each self-similar shape belongs to a one-parameter family determined by a permutation verifying an algebraic equation and for each integer $n$, there exists $\phi(n)$ such shapes, $\phi$ being the Euler totient function. In each family, the parameter for which two sides of the self-similar shape become tangent is calculated: it corresponds to a change of the number of self-intersection of the contour. Plant cell wall does not self-intersect thus the only contours which are biologically meaningful are the contour without self-intersection; the parameter of tangency thus give a rigorous and analytical estimate of the limit tissue curvature for which a given self-similar cell can be observed. The results could easily be generalized to self-similar asymmetric divisions minimizing the added perimeter (\textit{i.e} the ratio between the area of the daughter cell area and the mother cell area is $\lambda$ which can be different from $2$) by systematically changing the ratio $2$ in the formula by $\lambda$. 
\\\\
$\mathbf{Acknowledgement}$\\
The author thanks Jacques Dumais for providing this research topic.

\section{Annex}

\begin{lemma}\label{lm:casbase2}
Let $N>0$ an integer, let $\mathcal{P}$ a $N$-sided CCAOC on a cone of planar angular span $\gamma$ constituted by the $N$ CCAs, $(\mathcal{A}_{j})_{j\in\bZ / N \bZ}$, whose parameters are $(R_j,D_j,\phi_j)_{j\in \bZ / N \bZ}$.\\
Let suppose there exists $\rho>0$ such as:
\begin{eqnarray}
\forall j\in\bZ / N \bZ, \frac{R_j}{D_j}=\sqrt{\rho}.\label{Hypo1}
\end{eqnarray}
\begin{enumerate}
\item For any $j$ in $\bZ / N \bZ$,
\begin{eqnarray}
\delta\phi^*_j=\epsilon_j\arccos\left((1-\rho)\tau_{j,j+1}\right)\label{equ}
\end{eqnarray}
with:
\begin{eqnarray}
\tau_{j,j+1}=\frac{R_j^2+R_{j+1}^2}{2R_jR_{j+1}}.\label{tau_tan}
\end{eqnarray}
and $$(\epsilon_j)_{j\in\bZ/N\bZ}\in\{-1,1\}^N.$$
\item Let suppose for some $i\in\bZ / N \bZ$, $\mathcal{A}_{i}$ is tangent to $\mathcal{A}_{i+2}$; it implies $\rho$ belongs to 
$$\{0,1-\frac{\varepsilon_{i,i+2}}{\tau_{i,i+1}\tau_{i+1,i+2}},1\}$$
 with $\varepsilon_{i,i+2}=1$ in case of external tangency or $\varepsilon_{i,i+2}=-1$ in case of internal tangency. Each value of $\rho$ corresponds to determined values of $(\delta\phi_{i},\delta\phi_{i+1})$:
\begin{itemize}
\item To 
\begin{eqnarray}
\rho=0 \label{val_rho_0}
\end{eqnarray}
 corresponds: 
\begin{eqnarray}
\delta\phi^*_i&=&\epsilon_i\arccos\left(\tau_{i,i+1}\right)\label{rootrho01}\\
\delta\phi^*_{i+1}&=&\epsilon_{i+1}\arccos\left(\tau_{i+1,i+2}\right)\label{rootrho02}
\end{eqnarray}
with $\epsilon_{i+1}=\epsilon_i$. \\
\item To 
\begin{eqnarray}
\rho=1-\frac{\varepsilon_{i,i+2}}{\tau_{i,i+1}\tau_{i+1,i+2}} \label{val_rho_int}
\end{eqnarray}
 corresponds two real values:
\begin{eqnarray}
\delta\phi^*_{i}&=&\epsilon_i\arccos\Big(\frac{\varepsilon_{i,i+2}}{\tau_{i+1,i+2}}\Big)\label{rootdphi1}\\
\delta\phi^*_{i+1}&=&\epsilon_{i+1} \arccos\Big(\frac{\varepsilon_{i,i+2}}{\tau_{i,i+1}}\Big)\label{rootdphi2}
\end{eqnarray}
with:
\begin{eqnarray}
\epsilon_i=sign\left(\varepsilon_{i,i+2}\left(R_{i+1}^2-R_i^2\right)\left(R_{i+1}^2-R_{i+2}^2\right)\right)\epsilon_{i+1}.
\end{eqnarray}
\item To 
\begin{eqnarray}
\rho=1\label{val_rho_1}
\end{eqnarray}
corresponds: 
\begin{eqnarray}
\delta\phi^*_{i}&=&\epsilon_i\frac{\pi}{2}\label{rootdphi01}\\
\delta\phi^*_{i+1}&=&\epsilon_{i+1}\frac{\pi}{2}\label{rootdphi02}
\end{eqnarray}
with $\epsilon_{i+1}\varepsilon_{i,i+2}=\epsilon_i$.
\end{itemize}
\item There exists a positive integer $m$ such as the planar angular span $\gamma$ reads:
\begin{eqnarray}
\gamma=\sum_{j\in\bZ/N\bZ}\left((1-\epsilon_j)\pi+\delta\phi^*_j\right)+2m\pi.\label{Angular_defect}
\end{eqnarray}
\end{enumerate}
\end{lemma}
\begin{proof}
\begin{enumerate}
\item Let consider $\mathcal{P}^{*}$. For any $j$, the condition of orthogonality between $\mathcal{A}^*_{j}$ and $\mathcal{A}^*_{j+1}$ reads:
\begin{eqnarray*}
(D_{j+1}\cos(\delta\phi_j)-D_j)^2+(D_{j+1}\sin(\delta\phi_j))^2&=&R_{j+1}^2+R_{j}^2.
\end{eqnarray*}
Once expanded, it reads:
\begin{eqnarray*}
D_{j+1}^2\cos(\delta\phi_j)^2-2D_jD_{j+1}\cos(\delta\phi_j)+D_j^2+D_{j+1}^2\sin(\delta\phi_j)^2&=&R_{j+1}^2+R_{j}^2.
\end{eqnarray*}
Gathering the squared trigonometric terms, the expression simplifies:
\begin{eqnarray*}
D_{j+1}^2+D_j^2-2D_jD_{j+1}\cos(\delta\phi_j)&=&R_{j+1}^2+R_{j}^2.
\end{eqnarray*}
Factorizing on the left by $D_j^2$ and on the right by $R_j^2$ gives:
\begin{eqnarray*}
D_j^2\Big(\Big(\frac{D_{j+1}}{D_j}\Big)^2+1-2\Big(\frac{D_{j+1}}{D_j}\Big)\cos(\delta\phi_j)\Big)&=&R_j^2\Big(1+\Big(\frac{R_{j+1}}{R_j}\Big)^2\Big)
\end{eqnarray*}
which can be rewritten:
\begin{eqnarray*}
D_j^2\Big(\Big(\frac{R_{j+1}}{R_j}\Big)^2+1-2\Big(\frac{R_{j+1}}{R_j}\Big)\cos(\delta\phi_j)\Big)&=&R_j^2\Big(1+\Big(\frac{R_{j+1}}{R_j}\Big)^2\Big)
\end{eqnarray*}
as (\ref{Hypo1}) implies: $\frac{D_{j+1}}{D_j}=\frac{R_{j+1}}{R_j}$.\\
Dividing both sides by $D_j^2\Big(1+\Big(\frac{R_{j+1}}{R_j}\Big)^2\Big)$ gives:
\begin{eqnarray*}
1-2\frac{R_{j+1}}{R_j}\frac{\cos(\delta\phi_j)}{\Big(1+\Big(\frac{R_{j+1}}{R_j}\Big)^2\Big)}&=&\frac{R_j^2}{D_j^2}
\end{eqnarray*}
which can be rewritten:
\begin{eqnarray*}
1-2\frac{R_jR_{j+1}}{R_j^2+R_{j+1}^2}\cos(\delta\phi_j)&=&\rho
\end{eqnarray*}
using notation (\ref{Hypo1}).
Finally, $\cos(\delta\phi_j)$ reads:
\begin{eqnarray}
\cos(\delta\phi_j)&=&(1-\rho)\frac{R_j^2+R_{j+1}^2}{2R_jR_{j+1}}\label{equm1}
\end{eqnarray}
and
\begin{eqnarray*}
\cos(\delta\phi_j)&=&(1-\rho)\tau_{j,j+1}
\end{eqnarray*}
which is equivalent to (\ref{equ}).
\item The tangency between $\mathcal{A}^*_{i}$ and $\mathcal{A}^*_{i+2}$ implies:
\begin{eqnarray*}
(D_{i+2}\cos(\delta\phi_i+\delta\phi_{i+1}))-D_i)^2+(D_{i+2}\sin(\delta\phi_i+\delta\phi_{i+1}))^2=(R_i+\varepsilon_{i,i+2}R_{i+2})^2
\end{eqnarray*}
\begin{eqnarray}
\label{tangency}
\end{eqnarray}
(\ref{tangency}) expands into:
\begin{eqnarray*}
D_{i+2}^2\cos(\delta\phi_i+\delta\phi_{i+1})^2&+&D_i^2-2D_iD_{i+2}\cos(\delta\phi_i+\delta\phi_{i+1}))+D_{i+2}^2\sin(\delta\phi_i+\delta\phi_{i+1})^2\\
&=&(R_i+\varepsilon_{i,i+2}R_{i+2})^2.
\end{eqnarray*}
Gathering the squared trigonometric terms simplifies the expression:
\begin{eqnarray*}
D_{i+2}^2+D_i^2-2D_{i+2}D_i\cos(\delta\phi_i+\delta\phi_{i+1})&=&(R_{i}+\varepsilon_{i,i+2}R_{i+2})^2.
\end{eqnarray*}
Substituting $\Big(\frac{D_{i}}{R_{i}}\Big)R_{i}$ to $D_{i}$ and $\Big(\frac{D_{i+2}}{R_{i+2}}\Big)R_{i+2}$ to $D_{i+2}$ gives:
\begin{eqnarray*}
\Big(\Big(\frac{D_{i+2}}{R_{i+2}}\Big)^2R_{i+2}^2&+&\Big(\frac{D_i}{R_i}\Big)^2R_i^2-2\Big(\frac{D_{i+2}}{R_{i+2}}\Big)R_{i+2}\Big(\frac{D_i}{R_i}\Big)R_i\cos(\delta\phi_i+\delta\phi_{i+1})\Big)
=(R_{i}+\varepsilon_{i,i+2}R_{i+2})^2
\end{eqnarray*}
\begin{eqnarray}
\label{etape}
\end{eqnarray}
(\ref{Hypo1}) implies $\frac{D_i}{R_i}=\frac{D_{i+2}}{R_{i+2}}$, thus (\ref{etape}) can be rewritten:
\begin{eqnarray*}
\Big(\Big(\frac{D_{i}}{R_{i}}\Big)^2R_{i+2}^2+\Big(\frac{D_i}{R_i}\Big)^2R_i^2-2\Big(\frac{D_i}{R_i}\Big)^2R_{i+2}R_i\cos(\delta\phi_i+\delta\phi_{i+1})\Big)&=&(R_{i}+\varepsilon_{i,i+2}R_{i+2})^2.
\end{eqnarray*}
The whole expression can factorized by $\Big(\frac{D_i}{R_i}\Big)^2$ which gives:
\begin{eqnarray*}
\Big(\frac{D_{i}}{R_{i}}\Big)^2(R_{i+2}^2+R_i^2-2R_iR_{i+2}\cos(\delta\phi_i+\delta\phi_{i+1}))&=&(R_{i}+\varepsilon_{i,i+2}R_{i+2})^2.
\end{eqnarray*}
As $R_{i+2}^2+R_i^2$ is equal to $(R_{i+2}+\varepsilon_{i,i+2} R_i)^2-2\varepsilon_{i,i+2}R_{i}R_{i+2}$, the expression can be rewritten as:
\begin{eqnarray*}
\Big(\frac{D_i}{R_i}\Big)^2((R_{i+2}+\varepsilon_{i,i+2} R_i)^2-2\varepsilon_{i,i+2}R_{i}R_{i+2}-2R_{i+2}R_i\cos(\delta\phi_i+\delta\phi_{i+1}))&=&(R_{i}+\varepsilon_{i,i+2}R_{i+2})^2.
\end{eqnarray*}
Dividing both sides by $\Big(\frac{D_i}{R_i}\Big)^2(R_{i}+\varepsilon_{i,i+2}R_{i+2})^2$ gives:
\begin{eqnarray*}
1-\frac{2R_{i+2}R_i}{(R_{i}+\varepsilon_{i,i+2}R_{i+2})^2}(\cos(\delta\phi_i+\delta\phi_{i+1})+\varepsilon_{i,i+2}))&=&\Big(\frac{R_i}{D_i}\Big)^2.
\end{eqnarray*}
Finally $\cos(\delta\phi_i+\delta\phi_{i+1})$ can be expressed as:
\begin{eqnarray*}
\cos(\delta\phi_i+\delta\phi_{i+1})&=&(1-\rho)\frac{(R_i+\varepsilon_{i,i+2}R_{i+2})^2}{2R_{i}R_{i+2}}-\varepsilon_{i,i+2}
\end{eqnarray*}
which can be rewritten:
\begin{eqnarray}
\cos(\delta\phi_i+\delta\phi_{i+1})&=&(1-\rho)\tau_{tan,i,i+2}-\varepsilon_{i,i+2}\label{eau}
\end{eqnarray}
using the notation:
\begin{eqnarray}
\tau_{tan,i,i+2}=\frac{(R_i+\varepsilon_{i,i+2}R_{i+2})^2}{2R_{i}R_{i+2}}.\label{Formula_tau_tan}
\end{eqnarray}
The left side of (\ref{eau}) can be expanded:
\begin{eqnarray}
\cos(\delta\phi_i)\cos(\delta\phi_{i+1})-\sin(\delta\phi_i)\sin(\delta\phi_{i+1})&=&X\tau_{tan,i,i+2}-\varepsilon_{i,i+2}\label{eqo}
\end{eqnarray}
with $X=(1-\rho)$.
substituting using (\ref{equ}) gives:
\begin{eqnarray*}
X^2\tau_{i,i+1}\tau_{i+1,i+2}-\sin(\delta\phi_i)\sin(\delta\phi_{i+1})&=&X\tau_{tan,i,i+2}-\varepsilon_{i,i+2}
\end{eqnarray*}
which rearranges into:
\begin{eqnarray*}
X^2\tau_{i,i+1}\tau_{i+1,i+2}-X\tau_{tan,i,i+2}+\varepsilon_{i,i+2}&=&\sin(\delta\phi_i)\sin(\delta\phi_{i+1}).
\end{eqnarray*}
We raise to the square both sides:
\begin{eqnarray*}
(X^2\tau_{i,i+1}\tau_{i+1,i+2}-X\tau_{tan,i,i+2}+\varepsilon_{i,i+2})^2&=&(1-\cos^2(\delta\phi_i))(1-\cos^2(\delta\phi_{i+1})).
\end{eqnarray*}
Substituting using (\ref{equ}) gives:
\begin{eqnarray*}
(X^2\tau_{i,i+1}\tau_{i+1,i+2}-X\tau_{tan,i,i+2}+\varepsilon_{i,i+2})^2&=&(1-X^2\tau_{i,i+1}^2)(1-X^2\tau_{i+1,i+2}^2))
\end{eqnarray*}
and once the parenthesis expanded:
\begin{eqnarray*}
X^4\tau_{i,i+1}^2\tau_{i+1,i+2}^2&+&X^2\tau_{tan,i,i+2}^2+\varepsilon_{i,i+2}^2-2X^3\tau_{i,i+1}\tau_{i+1,i+2}\tau_{tan,i,i+2}+2X^2\tau_{i,i+1}\tau_{i+1,i+2}\varepsilon_{i,i+2}\\
&-&2X\tau_{tan,i,i+2}\varepsilon_{i,i+2}=1-X^2\tau_{i,i+1}^2-X^2\tau_{i+1,i+2}^2+X^4\tau_{i,i+1}^2\tau_{i+1,i+2}^2.
\end{eqnarray*}
As $\varepsilon_{i,i+2}=\pm1$, it rewrites:
\begin{eqnarray*}
X^4\tau_{i,i+1}^2\tau_{i+1,i+2}^2&+&X^2\tau_{tan,i,i+2}^2+1-2X^3\tau_{i,i+1}\tau_{i+1,i+2}\tau_{tan,i,i+2}+2X^2\tau_{i,i+1}\tau_{i+1,i+2}\varepsilon_{i,i+2}\\
&-&2X\tau_{tan,i,i+2}\varepsilon_{i,i+2}=1-X^2\tau_{i,i+1}^2-X^2\tau_{i+1,i+2}^2+X^4\tau_{i,i+1}^2\tau_{i+1,i+2}^2.
\end{eqnarray*}
The $0$ order and $4$ order monomials on both sides cancel each other:
\begin{eqnarray*}
X^2\tau_{tan,i,i+2}^2&+&2\varepsilon_{i,i+2}X^2\tau_{i,i+1}\tau_{i+1,i+2}-2\varepsilon_{i,i+2}X\tau_{tan,i,i+2}-2X^3\tau_{i,i+1}\tau_{i+1,i+2}\tau_{tan,i,i+2}\\
&=&-X^2\tau_{i,i+1}^2-X^2\tau_{i+1,i+2}^2.
\end{eqnarray*}
\begin{eqnarray}
\label{third_degree_polynomial}
\end{eqnarray}
$X$ is thus a root of a third degree polynomial and it takes at most three values; as both sides can be divided by $X$, $0$ is one of these roots.
substituting the root $X=0$ into (\ref{equ}) gives:
\begin{eqnarray*}
\delta\phi^*_{i}&=&\epsilon_i\arccos(0),\\
\delta\phi^*_{i+1}&=&\epsilon_{i+1}\arccos(0)
\end{eqnarray*}
thus providing the relations:
\begin{eqnarray}
\delta\phi^*_{i}&=&\epsilon_i\frac{\pi}{2},\label{oubouloulou1}\\
\delta\phi^*_{i+1}&=&\epsilon_{i+1}\frac{\pi}{2}.\label{oubouloulou2}
\end{eqnarray}
Substituting $X=0$ and (\ref{oubouloulou1},\ref{oubouloulou2}) into (\ref{eqo}) gives:
\begin{eqnarray*}
0-\epsilon_i\epsilon_{i+1}&=&0-\varepsilon_{i,i+2}.
\end{eqnarray*}
As $$ | \epsilon_i|=| \epsilon_{i+1}|=|\varepsilon_{i,i+2}|=1,$$
it is equivalent to
\begin{eqnarray*}
 \epsilon_{i+1}\varepsilon_{i,i+2}=\epsilon_i
\end{eqnarray*}
thus finishing to prove (\ref{rootdphi01},\ref{rootdphi02}).\\\\
Dividing by X the third degree polynomial (\ref{third_degree_polynomial}) gives:
\begin{eqnarray*}
X\tau_{tan,i,i+2}^2&+&2\varepsilon_{i,i+2}X\tau_{i,i+1}\tau_{i+1,i+2}-2\varepsilon_{i,i+2}\tau_{tan,i,i+2}-2X^2\tau_{i,i+1}\tau_{i+1,i+2}\tau_{tan,i,i+2}\\
&=&-X\tau_{i,i+1}^2-X\tau_{i+1,i+2}^2
\end{eqnarray*}
whose terms can be regrouped into:
\begin{eqnarray*}
-2\varepsilon_{i,i+2}\tau_{tan,i,i+2}&+&X(\tau_{i,i+1}^2+2\varepsilon_{i,i+2}\tau_{i,i+1}\tau_{i+1,i+2}+\tau_{i+1,i+2}^2
+\tau_{tan,i,i+2}^2)\\&-&2X^2\tau_{i,i+1}\tau_{i+1,i+2}\tau_{tan,i,i+2}=0.
\end{eqnarray*}
The two remaining roots of (\ref{third_degree_polynomial}) also cancel the second degree polynomial:
\begin{eqnarray*}
-2\varepsilon_{i,i+2}\tau_{tan,i,i+2}+X((\tau_{i,i+1}+\varepsilon_{i,i+2}\tau_{i+1,i+2})^2+\tau_{tan,i,i+2}^2)-2X^2\tau_{i,i+1}\tau_{i+1,i+2}\tau_{tan,i,i+2}&=&0.
\end{eqnarray*}
\begin{eqnarray}
\label{ogu}
\end{eqnarray}
Let suppose $X=1$. Substituting into (\ref{equ}) gives:
\begin{eqnarray*}
\delta\phi^*_i&=&\epsilon_i\arccos\left(\tau_{i,i+1}\right),\\
\delta\phi^*_{i+1}&=&\epsilon_{i+1}\arccos\left(\tau_{i+1,i+2}\right).
\end{eqnarray*}
Moreover substituting $X=1$ on the right side of (\ref{eqo}) gives:
\begin{eqnarray*}
\frac{(R_i+\varepsilon_{i,i+2}R_{i+2})^2}{2R_{i}R_{i+2}}-\varepsilon_{i,i+2}=\frac{R_i^2+2\varepsilon_{i,i+2}R_iR_{i+2}+\varepsilon_{i,i+2}^2R_{i+2}^2}{2R_{i}R_{i+2}}-\frac{\varepsilon_{i,i+2}2R_{i}R_{i+2}}{2R_{i}R_{i+2}}
\end{eqnarray*}
which simplifies as $\varepsilon_{i,i+2}=\pm1$ into:
\begin{eqnarray}
\frac{(R_i+\varepsilon_{i,i+2}R_{i+2})^2}{2R_{i}R_{i+2}}-\varepsilon_{i,i+2}=\frac{R_i^2+R_{i+2}^2}{2R_{i}R_{i+2}}.\label{lefteqouu}
\end{eqnarray}
For a given $j$, $\tau_{j,j+1}$ can be rewritten:
\begin{eqnarray*}
\tau_{j,j+1}=\frac{\frac{R_j^2}{R_jR_{j+1}}+\frac{R_{j+1}^2}{R_jR_{j+1}}}{2}
\end{eqnarray*}
which simplifies into:
\begin{eqnarray*}
\tau_{j,j+1}=\frac{\frac{R_j}{R_{j+1}}+\frac{R_{j+1}}{R_j}}{2}
\end{eqnarray*}
and finally rewrites
\begin{eqnarray}
\tau_{j,j+1}=\frac{\kappa_j+\frac{1}{\kappa_j}}{2}\label{express}
\end{eqnarray}
with the notation:
\begin{eqnarray}
\kappa_j=\frac{R_j}{R_{j+1}}.\label{okolo}
\end{eqnarray}
(\ref{equ}) implies:
\begin{eqnarray*}
\sin(\delta\phi_j)&=&\epsilon_j\sqrt{1-\cos^2(\delta\phi_j)}.
\end{eqnarray*}
Substituting with (\ref{express}) gives:
\begin{eqnarray}
\sin(\delta\phi_j)&=&\epsilon_j\sqrt{1-\Big(\frac{\kappa_j+\frac{1}{\kappa_j}}{2}\Big)^2}.\label{rooootoo}
\end{eqnarray}
The expression inside the root symbol can be expanded:
\begin{eqnarray*}
1-\Big(\frac{\kappa_j^2+\frac{1}{\kappa_j^2}+2\frac{\kappa_j}{\kappa_j}}{4}\Big)
\end{eqnarray*}
which can be rearranged as:
\begin{eqnarray*}
\frac{4-\kappa_j^2-\frac{1}{\kappa_j^2}-2}{4}
\end{eqnarray*}
simplifying into:
\begin{eqnarray*}
\frac{2-\kappa_j^2-\frac{1}{\kappa_j^2}}{4}
\end{eqnarray*}
which factorizes in a squared term:
\begin{eqnarray*}
\frac{-(\kappa_j-\frac{1}{\kappa_j})^2}{4}.
\end{eqnarray*}
Substituting inside the root symbol of (\ref{rooootoo}) gives:
\begin{eqnarray}
\sin(\delta\phi_j)&=& \epsilon_j\textbf{i} (\frac{\kappa_j-\frac{1}{\kappa_j}}{2}).\label{sinform}
\end{eqnarray}
Substituting (\ref{express},\ref{sinform}) in the left side of (\ref{eqo}) gives:
\begin{eqnarray*}
\cos(\delta\phi_i)\cos(\delta\phi_{i+1})-\sin(\delta\phi_i)\sin(\delta\phi_{i+1})&=&\frac{\kappa_i+\frac{1}{\kappa_i}}{2}\frac{\kappa_{i+1}+\frac{1}{\kappa_{i+1}}}{2}- \epsilon_i\epsilon_{i+1}\textbf{i}^2\frac{\kappa_i-\frac{1}{\kappa_i}}{2}\frac{\kappa_{i+1}-\frac{1}{\kappa_{i+1}}}{2}
\end{eqnarray*}
which simplifies into:
\begin{eqnarray*}
\cos(\delta\phi_i)\cos(\delta\phi_{i+1})-\sin(\delta\phi_i)\sin(\delta\phi_{i+1})&=&\frac{(\kappa_i+\frac{1}{\kappa_i})(\kappa_{i+1}+\frac{1}{\kappa_{i+1}})+\epsilon_i\epsilon_{i+1}(\kappa_i-\frac{1}{\kappa_i})(\kappa_{i+1}-\frac{1}{\kappa_{i+1}})}{4}.
\end{eqnarray*}
\begin{eqnarray}
\label{initooooooo}
\end{eqnarray}
Expanding the right side of the expression gives:
\begin{eqnarray*}
&&(\kappa_i+\frac{1}{\kappa_i})(\kappa_{i+1}+\frac{1}{\kappa_{i+1}})+ \epsilon_i\epsilon_{i+1}(\kappa_i-\frac{1}{\kappa_i})(\kappa_{i+1}-\frac{1}{\kappa_{i+1}})=\\
&&\kappa_i\kappa_{i+1}+\frac{\kappa_i}{\kappa_{i+1}}+\frac{\kappa_{i+1}}{\kappa_{i}}+\frac{1}{\kappa_{i}\kappa_{i+1}}+ \epsilon_i\epsilon_{i+1}\left(\kappa_i\kappa_{i+1}-\frac{\kappa_i}{\kappa_{i+1}}-\frac{\kappa_{i+1}}{\kappa_{i}}+\frac{1}{\kappa_{i}\kappa_{i+1}}\right).
\end{eqnarray*}
\begin{eqnarray}
\label{Developppo}
\end{eqnarray}
If $\epsilon_i=\epsilon_{i+1}$, (\ref{Developppo}) further simplifies in:
\begin{eqnarray*}
(\kappa_i+\frac{1}{\kappa_i})(\kappa_{i+1}+\frac{1}{\kappa_{i+1}})+(\kappa_i-\frac{1}{\kappa_i})(\kappa_{i+1}-\frac{1}{\kappa_{i+1}})=2\kappa_i\kappa_{i+1}+2\frac{1}{\kappa_{i}\kappa_{i+1}},
\end{eqnarray*}
\begin{eqnarray}
\label{simplifiyuu}
\end{eqnarray}
thus (\ref{initooooooo}) simplifies:
\begin{eqnarray}
\cos(\delta\phi_i)\cos(\delta\phi_{i+1})-\sin(\delta\phi_i)\sin(\delta\phi_{i+1})&=&\frac{\kappa_{i+1}\kappa_i+\frac{1}{\kappa_i\kappa_{i+1}}}{2}\label{simplif1}
\end{eqnarray}
else if $\epsilon_i=-\epsilon_{i+1}$, (\ref{Developppo}) further simplifies in:
\begin{eqnarray}
\cos(\delta\phi_i)\cos(\delta\phi_{i+1})-\sin(\delta\phi_i)\sin(\delta\phi_{i+1})&=&\frac{\frac{\kappa_{i+1}}{\kappa_i}+\frac{\kappa_i}{\kappa_{i+1}}}{2}\label{simplif2}
\end{eqnarray}
If $\epsilon_i=\epsilon_{i+1}$, substituting in (\ref{simplif1}) with (\ref{okolo}) gives:
\begin{eqnarray*}
\cos(\delta\phi_i)\cos(\delta\phi_{i+1})-\sin(\delta\phi_i)\sin(\delta\phi_{i+1})&=&\frac{\frac{R_{i+2}}{R_{i+1}}\frac{R_{i+1}}{R_{i}}+\frac{R_{i}}{R_{i+1}}\frac{R_{i+1}}{R_{i+2}}}{2}.
\end{eqnarray*}
which simplifies into:
\begin{eqnarray*}
\cos(\delta\phi_i)\cos(\delta\phi_{i+1})-\sin(\delta\phi_i)\sin(\delta\phi_{i+1})&=&\frac{\frac{R_{i+2}}{R_{i}}+\frac{R_{i}}{R_{i+2}}}{2}
\end{eqnarray*}
and can be rearranged as:
\begin{eqnarray}
\cos(\delta\phi_i)\cos(\delta\phi_{i+1})-\sin(\delta\phi_i)\sin(\delta\phi_{i+1})&=&\frac{R_{i+2}^2+R_{i}^2}{2R_{i+2}R_i}.\label{glopmmmm}
\end{eqnarray}
(\ref{glopmmmm}) equals (\ref{lefteqouu}): $X=1$ (i.e $\rho=0$) is a root of (\ref{ogu}) which corresponds to the case $\epsilon_i=\epsilon_{i+1}$ thus proving (\ref{rootrho01},\ref{rootrho02}).\\
If $\epsilon_i=-\epsilon_{i+1}$, substituting in (\ref{simplif2}) with (\ref{okolo}) gives:
\begin{eqnarray*}
\cos(\delta\phi_i)\cos(\delta\phi_{i+1})-\sin(\delta\phi_i)\sin(\delta\phi_{i+1})&=&\frac{\frac{R_{i+2}}{R_{i+1}}\frac{R_{i}}{R_{i+1}}+\frac{R_{i+1}}{R_{i}}\frac{R_{i+1}}{R_{i+2}}}{2}.
\end{eqnarray*}
which simplifies into:
\begin{eqnarray*}
\cos(\delta\phi_i)\cos(\delta\phi_{i+1})-\sin(\delta\phi_i)\sin(\delta\phi_{i+1})&=&\frac{\frac{R_{i}R_{i+2}}{R_{i+1}^2}+\frac{R_{i+1}^2}{R_{i}R_{i+2}}}{2}
\end{eqnarray*}
which is in general different from (\ref{lefteqouu}) thus $\epsilon_i=-\epsilon_{i+1}$ does not corresponds to a root canceling the equation (\ref{ogu}): the solutions (\ref{rootrho01},\ref{rootrho02}) are the only one corresponding to $\rho=0$ (i.e $X=1$).\\
As the ratio between the lower order term and the higher order term of the polynomial (\ref{ogu}) gives the product of both roots:
$$\frac{-2\varepsilon_{i,i+2}\tau_{tan,i,i+2}}{-2\tau_{i,i+1}\tau_{i+1,i+2}\tau_{tan,i,i+2}},$$
the other root reads: 
\begin{eqnarray*}
X&=&\frac{-2\varepsilon_{i,i+2} \tau_{tan,i,i+2}}{-2\tau_{i,i+1}\tau_{i+1,i+2}\tau_{tan,i,i+2}}
\end{eqnarray*}
which simplifies into:
\begin{eqnarray}
X&=&\frac{\varepsilon_{i,i+2} }{\tau_{i,i+1}\tau_{i+1,i+2}}.\label{root}
\end{eqnarray}
Substituting this value of $X$ into (\ref{equ}) gives:
\begin{eqnarray}
\delta\phi^*_{i}&=&\epsilon_i\arccos\Big(\frac{\varepsilon_{i,i+2}}{\tau_{i+1,i+2}}\Big),\label{rootdphi1X}\\
\delta\phi^*_{i+1}&=&\epsilon_{i+1} \arccos\Big(\frac{\varepsilon_{i,i+2}}{\tau_{i,i+1}}\Big).\label{rootdphi2X}
\end{eqnarray}
Substituting (\ref{rootdphi1X},\ref{rootdphi2X}) in (\ref{eqo}) gives:
\begin{eqnarray*}
&&\Big(\frac{\varepsilon_{i,i+2}}{\tau_{i+1,i+2}}\Big)\Big(\frac{\varepsilon_{i,i+2}}{\tau_{i,i+1}}\Big)\\
&-&\epsilon_i\sqrt{1-\Big(\frac{\varepsilon_{i,i+2}}{\tau_{i+1,i+2}}\Big)^2}\epsilon_{i+1}\sqrt{1-\Big(\frac{\varepsilon_{i,i+2}}{\tau_{i,i+1}}\Big)^2}\\
&=&\frac{\varepsilon_{i,i+2} }{\tau_{i,i+1}\tau_{i+1,i+2}}\tau_{tan,i,i+2}-\varepsilon_{i,i+2}.
\end{eqnarray*}
Substituting $\varepsilon_{i,i+2}=\pm1$ on the left side and factorizing by $\varepsilon_{i,i+2}$ on the right side, it rewrites:
\begin{eqnarray*}
&&\frac{1}{\tau_{i,i+1}\tau_{i+1,i+2}}-\sqrt{1-\frac{1}{\tau_{i+1,i+2}^2}}\epsilon_i\epsilon_{i+1}\sqrt{1-\frac{1}{\tau_{i,i+1}^2}}\\
&=&\varepsilon_{i,i+2} \left(\frac{\tau_{tan,i,i+2}}{\tau_{i,i+1}\tau_{i+1,i+2}}-1\right).
\end{eqnarray*}
Multiplying both sides by $\tau_{i,i+1}\tau_{i+1,i+2}$ gives:
\begin{eqnarray*}
1-\epsilon_i\epsilon_{i+1}\sqrt{\tau_{i,i+1}^2-1}\sqrt{\tau_{i+1,i+2}^2-1}&=&\varepsilon_{i,i+2}(\tau_{tan,i,i+2}-\tau_{i,i+1}\tau_{i+1,i+2})
\end{eqnarray*}
which rewrite:
\begin{eqnarray*}
1&-&\epsilon_i\epsilon_{i+1}\sqrt{\left(\frac{\frac{R_i}{R_{i+1}}+\frac{R_{i+1}}{R_{i}}}{2}\right)^2-1}\sqrt{\left(\frac{\frac{R_{i+1}}{R_{i+2}}+\frac{R_{i+2}}{R_{i+1}}}{2}\right)^2-1}\\
&=&\varepsilon_{i,i+2}\left(\frac{(R_i+\varepsilon_{i,i+2}R_{i+2})^2}{2R_{i}R_{i+2}}-\left(\frac{\frac{R_i}{R_{i+1}}+\frac{R_{i+1}}{R_{i}}}{2}\right)\left(\frac{\frac{R_{i+1}}{R_{i+2}}+\frac{R_{i+2}}{R_{i+1}}}{2}\right)\right).
\end{eqnarray*}
\begin{eqnarray}
\label{gulgulu}
\end{eqnarray}
For a given $j$, the  binomial expansion gives:
\begin{eqnarray*}
\left(\frac{\frac{R_j}{R_{j+1}}+\frac{R_{j+1}}{R_{j}}}{2}\right)^2=\left(\frac{\frac{R_j}{R_{j+1}}}{2}\right)^2+\left(\frac{\frac{R_{j+1}}{R_{j}}}{2}\right)^2+2\left(\frac{\frac{R_j}{R_{j+1}}}{2}\right)\left(\frac{\frac{R_{j+1}}{R_{j}}}{2}\right)
\end{eqnarray*}
which simplifies into:
\begin{eqnarray*}
\left(\frac{\frac{R_j}{R_{j+1}}+\frac{R_{j+1}}{R_{j}}}{2}\right)^2=\left(\frac{\frac{R_j}{R_{j+1}}}{2}\right)^2+\left(\frac{\frac{R_{j+1}}{R_{j}}}{2}\right)^2+\frac{1}{2}.
\end{eqnarray*}
Substracting $1$ on both sides give:
\begin{eqnarray*}
\left(\frac{\frac{R_j}{R_{j+1}}+\frac{R_{j+1}}{R_{j}}}{2}\right)^2-1=\left(\frac{\frac{R_j}{R_{j+1}}}{2}\right)^2+\left(\frac{\frac{R_{j+1}}{R_{j}}}{2}\right)^2-\frac{1}{2}
\end{eqnarray*}
which can be rewritten:
\begin{eqnarray*}
\left(\frac{\frac{R_j}{R_{j+1}}+\frac{R_{j+1}}{R_{j}}}{2}\right)^2-1=\left(\frac{\frac{R_j}{R_{j+1}}}{2}-\frac{\frac{R_{j+1}}{R_{j}}}{2}\right)^2
\end{eqnarray*}
thus:
\begin{eqnarray*}
\sqrt{\left(\frac{\frac{R_j}{R_{j+1}}+\frac{R_{j+1}}{R_{j}}}{2}\right)^2-1}=\left|\frac{\frac{R_j}{R_{j+1}}-\frac{R_{j+1}}{R_{j}}}{2}\right|.
\end{eqnarray*}
which can be rewritten:
\begin{eqnarray}
\sqrt{\left(\frac{\frac{R_j}{R_{j+1}}+\frac{R_{j+1}}{R_{j}}}{2}\right)^2-1}=\left|\frac{R_j^2-R_{j+1}^2}{2R_{j}R_{j+1}}\right|.\label{rootuio}
\end{eqnarray}
Using the square root terms on the left side of (\ref{gulgulu}) rewrites :
\begin{eqnarray*}
&&\epsilon_i\epsilon_{i+1}\sqrt{\left(\frac{\frac{R_i}{R_{i+1}}+\frac{R_{i+1}}{R_{i}}}{2}\right)^2-1}\sqrt{\left(\frac{\frac{R_{i+1}}{R_{i+2}}+\frac{R_{i+2}}{R_{i+1}}}{2}\right)^2-1}\\&=&\epsilon_i\epsilon_{i+1}\left|\frac{R_i^2-R_{i+1}^2}{2R_{i}R_{i+1}}\right|\left|\frac{R_{i+1}^2-R_{i+2}^2}{2R_{i+1}R_{j+2}}\right|.
\end{eqnarray*}
The left side of (\ref{gulgulu}) simplifies into:
\begin{eqnarray*}
1&-&\epsilon_i\epsilon_{i+1}\sqrt{\left(\frac{\frac{R_i}{R_{i+1}}+\frac{R_{i+1}}{R_{i}}}{2}\right)^2-1}\sqrt{\left(\frac{\frac{R_{i+1}}{R_{i+2}}+\frac{R_{i+2}}{R_{i+1}}}{2}\right)^2-1}\\
&=&1-\epsilon_i\epsilon_{i+1}\left|\frac{R_i^2-R_{i+1}^2}{2R_{i}R_{i+1}}\right|\left|\frac{R_{i+1}^2-R_{i+2}^2}{2R_{i+1}R_{j+2}}\right|.
\end{eqnarray*}
\begin{eqnarray}
\label{leftsidesimp}
\end{eqnarray}
Putting the same denominator to the two fractions on the right side (\ref{gulgulu}) gives:
\begin{eqnarray*}
&\ &\frac{(R_i+\varepsilon_{i,i+2}R_{i+2})^2}{2R_{i}R_{i+2}}-\left(\frac{R_i^2+R_{i+1}^2}{2R_iR_{i+1}}\right)\left(\frac{R_{i+1}^2+R_{i+2}^2}{2R_{i+1}R_{i+2}}\right)\\
&=&\frac{2(R_i+\varepsilon_{i,i+2}R_{i+2})^2R_{i+1}^2-\left(R_i^2+R_{i+1}^2\right)\left(R_{i+1}^2+R_{i+2}^2\right)}{4R_iR_{i+1}^2R_{i+2}}.
\end{eqnarray*}
\begin{eqnarray}
\label{rightside1}
\end{eqnarray}
Expanding the numerator on the right side of (\ref{rightside1}) gives:
\begin{eqnarray*}
&&2(R_i+\varepsilon_{i,i+2}R_{i+2})^2R_{i+1}^2-\left(R_i^2+R_{i+1}^2\right)\left(R_{i+1}^2+R_{i+2}^2\right)\\
&=&2R_i^2R_{i+1}^2+2\varepsilon_{i,i+2}^2R_{i+1}^2R_{i+2}^2+4\varepsilon_{i,i+2}R_iR_{i+1}^2R_{i+2}-R_i^2R_{i+1}^2-R_{i+1}^4-R_{i+1}^2R_{i+2}^2
\end{eqnarray*}
substituting $\varepsilon_{i,i+2}=\pm1$ gives:
\begin{eqnarray*}
&&2(R_i+\varepsilon_{i,i+2}R_{i+2})^2R_{i+1}^2-\left(R_i^2+R_{i+1}^2\right)\left(R_{i+1}^2+R_{i+2}^2\right)\\
&=&2R_i^2R_{i+1}^2+2R_{i+1}^2R_{i+2}^2+4\varepsilon_{i,i+2}R_iR_{i+1}^2R_{i+2}-R_i^2R_{i+1}^2-R_{i+1}^4-R_{i+1}^2R_{i+2}^2.
\end{eqnarray*}
Terms can be regrouped:
\begin{eqnarray*}
&&2(R_i+\varepsilon_{i,i+2}R_{i+2})^2R_{i+1}^2-\left(R_i^2+R_{i+1}^2\right)\left(R_{i+1}^2+R_{i+2}^2\right)\\
&=&4\varepsilon_{i,i+2}R_iR_{i+1}^2R_{i+2}+R_i^2R_{i+1}^2-R_{i+1}^4+R_{i+1}^2R_{i+2}^2
\end{eqnarray*}
which finally simplifies into:
\begin{eqnarray*}
&&2(R_i+\varepsilon_{i,i+2}R_{i+2})^2R_{i+1}^2-\left(R_i^2+R_{i+1}^2\right)\left(R_{i+1}^2+R_{i+2}^2\right)\\
&=&4\varepsilon_{i,i+2}R_iR_{i+1}^2R_{i+2}-(R_{i+1}^2-R_i^2)(R_{i+1}^2-R_{i+2}^2).
\end{eqnarray*}
\begin{eqnarray}
\label{detaillloooooo}
\end{eqnarray}
Substituting (\ref{detaillloooooo}) on the right side of (\ref{gulgulu}) gives:
\begin{eqnarray*}
&&\varepsilon_{i,i+2}\frac{4\varepsilon_{i,i+2}R_iR_{i+1}^2R_{i+2}+\left(R_i^2-R_{i+1}^2\right)\left(R_{i+1}^2-R_{i+2}^2\right)}{4R_iR_{i+1}^2R_{i+2}}\\
&=&\varepsilon_{i,i+2}\frac{4\varepsilon_{i,i+2}R_iR_{i+1}^2R_{i+2}-(R_{i+1}^2-R_i^2)(R_{i+1}^2-R_{i+2}^2)}{4R_iR_{i+1}^2R_{i+2}}
\end{eqnarray*}
and:
\begin{eqnarray*}
&&\varepsilon_{i,i+2}\frac{4\varepsilon_{i,i+2}R_iR_{i+1}^2R_{i+2}+\left(R_i^2-R_{i+1}^2\right)\left(R_{i+1}^2-R_{i+2}^2\right)}{4R_iR_{i+1}^2R_{i+2}}\\
&=&\varepsilon_{i,i+2}^2-\varepsilon_{i,i+2}\frac{(R_{i+1}^2-R_i^2)(R_{i+1}^2-R_{i+2}^2)}{4R_iR_{i+1}^2R_{i+2}}
\end{eqnarray*}
finally as $\varepsilon_{i,i+2}=\pm1$:
\begin{eqnarray}
&&\frac{4\varepsilon_{i,i+2}^2R_iR_{i+1}^2R_{i+2}+\varepsilon_{i,i+2}\left(R_i^2-R_{i+1}^2\right)\left(R_{i+1}^2-R_{i+2}^2\right)}{4R_iR_{i+1}^2R_{i+2}}\\
&=&1-\varepsilon_{i,i+2}\frac{(R_{i+1}^2-R_i^2)(R_{i+1}^2-R_{i+2}^2)}{4R_iR_{i+1}^2R_{i+2}}.\label{Rightsimp}
\end{eqnarray}
Thus the left side (\ref{leftsidesimp}) and the right side (\ref{Rightsimp}) of (\ref{gulgulu}) can only be equal if:
 \begin{eqnarray*}
&&1-\epsilon_i\epsilon_{i+1}\left|\frac{R_i^2-R_{i+1}^2}{2R_{i}R_{i+1}}\right|\left|\frac{R_{i+1}^2-R_{i+2}^2}{2R_{i+1}R_{j+2}}\right|\\
&=&1-\varepsilon_{i,i+2}\frac{(R_{i+1}^2-R_i^2)(R_{i+1}^2-R_{i+2}^2)}{4R_iR_{i+1}^2R_{i+2}}
\end{eqnarray*}
implying:
\begin{eqnarray*}
\epsilon_{i+1}=\epsilon_isign\left(\varepsilon_{i,i+2}\left(R_{i+1}^2-R_i^2\right)\left(R_{i+1}^2-R_{i+2}^2\right)\right)
\end{eqnarray*}
thus finishing to prove (\ref{rootdphi1},\ref{rootdphi2}).
\item Let $\mathcal{P}^{*}$ the CCAOC mapped onto the plane; it is constituted by $N$ arcs of circle $(\mathcal{A}^{*}_{j})_{j\in\bZ/N\bZ}$ whose centers are noted $(C^{*}_{j})_{j\in\bZ/N\bZ}$ with coordinates $(D_j(\cos(\phi_j),\sin(\phi_j)))_{j\in\bZ/N\bZ}$. Let $j$ in $\bZ/N\bZ$ and consider the two orthogonal circles containing $\mathcal{A}_j^{*}$ and $\mathcal{A}_{j+1}^{*}$ around their intersection. The two circles are orthogonal thus the oriented angle between the radii oriented from the intersection toward the centers are equal to the oriented angles between the tangents:
\begin{eqnarray}
\left(\overrightarrow{P^{*}_{j+1}C^{*}_j},\overrightarrow{P^{*}_{j+1}C^{*}_{j+1}}\right)=\left(\mathcal{A}_j^\prime(\phi^-_{P_j}), \mathcal{A}_{j+1}^\prime(\phi^+_{P_j})\right),
\end{eqnarray}
substituting (\ref{orth90}) gives:
\begin{eqnarray}
\left(\overrightarrow{P^{*}_{j+1}C^{*}_j},\overrightarrow{P^{*}_{j+1}C^{*}_{j+1}}\right)=\frac{\pi}{2}.\label{angle_normal}
\end{eqnarray}
The sum of the interior angle of the quadrangle $C^{*}_jP^{*}_{j+1}C^{*}_{j+1}O$ reads:
\begin{eqnarray*}
\left(\overrightarrow{C^{*}_jO},\overrightarrow{C^{*}_jP^{*}_{j+1}}\right)_i+\left(\overrightarrow{P^{*}_{j+1}C^{*}_j},\overrightarrow{P^{*}_{j+1}C^{*}_{j+1}}\right)_i+\left(\overrightarrow{C^{*}_{j+1}P^{*}_{j+1}},\overrightarrow{C^{*}_{j+1}O}\right)_i+\left(\overrightarrow{OC^{*}_{j+1}},\overrightarrow{OC^{*}_j}\right)_i=2\pi.
\end{eqnarray*}
\begin{eqnarray}
\label{quadrangle}
\end{eqnarray}
Inside a non-intersecting polygon, all the oriented angles going between consecutive sides have the same sign:
\begin{itemize}
\item If $\epsilon_j=1$ then $\delta\phi^*_j>0$, in the quadrangle $C^{*}_jP^{*}_{j+1}C^{*}_{j+1}O$, $$\left(\overrightarrow{OC^{*}_{j+1}},\overrightarrow{OC^{*}_{j}}\right)=-\delta\phi^*_j,$$ it implies $\left(\overrightarrow{C^{*}_jO},\overrightarrow{C^{*}_jP^{*}_{j+1}}\right)$, $\left(\overrightarrow{C^{*}_{j+1}P^{*}_{j+1}},\overrightarrow{C^{*}_{j+1}O}\right)$ and $\left(\overrightarrow{P^{*}_{j+1}C^{*}_j},\overrightarrow{P^{*}_{j+1}C^{*}_{j+1}}\right)$ are also negative quantities. It means:
\begin{eqnarray}
\left(\overrightarrow{OC^{*}_{j+1}},\overrightarrow{OC^{*}_{j}}\right)_i&=&\left \lfloor{-\left(\overrightarrow{OC^{*}_{j+1}},\overrightarrow{OC^{*}_{j}}\right)}\right \rfloor_{[0,2\pi]},\label{angoquadra1}\\
\left(\overrightarrow{C^{*}_jO},\overrightarrow{C^{*}_jP^{*}_{j+1}}\right)_i&=&\left \lfloor{-\left(\overrightarrow{C^{*}_jO},\overrightarrow{C^{*}_jP^{*}_{j+1}}\right)}\right \rfloor_{[0,2\pi]},\label{angoquadra2}\\
\left(\overrightarrow{P^{*}_{j+1}C^{*}_j},\overrightarrow{P^{*}_{j+1}C^{*}_{j+1}}\right)_i&=&\left \lfloor{-\left(\overrightarrow{P^{*}_{j+1}C^{*}_j},\overrightarrow{P^{*}_{j+1}C^{*}_{j+1}}\right)}\right \rfloor_{[0,2\pi]},\label{angoquadra4}\\
\left(\overrightarrow{C^{*}_{j+1}P^{*}_{j+1}},\overrightarrow{C^{*}_{j+1}O}\right)_i&=&\left \lfloor{-\left(\overrightarrow{C^{*}_{j+1}P^{*}_{j+1}},\overrightarrow{C^{*}_{j+1}O}\right)}\right \rfloor_{[0,2\pi]}.\label{angoquadra3}
\end{eqnarray}
As $\delta\phi^*_j$ in $[0,\pi]$, (\ref{angoquadra1}) can be rewritten:
\begin{eqnarray}
\left(\overrightarrow{OC^{*}_{j+1}},\overrightarrow{OC^{*}_{j}}\right)_i&=&\delta\phi^*_j.\label{angoquadra1bon}
\end{eqnarray}
Substituting (\ref{angle_normal}) in (\ref{angoquadra4}) gives:
\begin{eqnarray}
\left(\overrightarrow{P^{*}_{j+1}C^{*}_j},\overrightarrow{P^{*}_{j+1}C^{*}_{j+1}}\right)_i=\frac{3\pi}{2}.\label{angoquadra4bon}
\end{eqnarray}
Using (\ref{angoquadra1bon}) and (\ref{angoquadra4bon}) in (\ref{quadrangle}) gives:
\begin{eqnarray*}
\left(\overrightarrow{C^{*}_jO},\overrightarrow{C^{*}_jP^{*}_{j+1}}\right)_i+\frac{3\pi}{2}+\left(\overrightarrow{C^{*}_{j+1}P^{*}_{j+1}},\overrightarrow{C^{*}_{j+1}O}\right)_i+\delta\phi^*_j=2\pi
\end{eqnarray*}
which rewrites:
\begin{eqnarray}
\left(\overrightarrow{C^{*}_jO},\overrightarrow{C^{*}_jP^{*}_{j+1}}\right)_i+\left(\overrightarrow{C^{*}_{j+1}P^{*}_{j+1}},\overrightarrow{C^{*}_{j+1}O}\right)_i=\frac{\pi}{2}-\delta\phi^*_j.\label{angoquadrabouou}
\end{eqnarray}
Both $\left(\overrightarrow{C^{*}_jO},\overrightarrow{C^{*}_jP^{*}_{j+1}}\right)_i$ and $\left(\overrightarrow{C^{*}_{j+1}P^{*}_{j+1}},\overrightarrow{C^{*}_{j+1}O}\right)_i$ are positive thus as their sum is inferior to $\pi$ they are both inferior to $2\pi$ which means (\ref{angoquadra2},\ref{angoquadra3}) can be rewritten:
\begin{eqnarray}
\left(\overrightarrow{C^{*}_jO},\overrightarrow{C^{*}_jP^{*}_{j+1}}\right)_i&=&-\left(\overrightarrow{C^{*}_jO},\overrightarrow{C^{*}_jP^{*}_{j+1}}\right)\label{angoquadrab2},\\
\left(\overrightarrow{C^{*}_{j+1}P^{*}_{j+1}},\overrightarrow{C^{*}_{j+1}O}\right)_i&=&-\left(\overrightarrow{C^{*}_{j+1}P^{*}_{j+1}},\overrightarrow{C^{*}_{j+1}O}\right).\label{angoquadrab3}
\end{eqnarray}
Substituting (\ref{angoquadrab2},\ref{angoquadrab3}) in (\ref{angoquadrabouou}) gives:
\begin{eqnarray*}
\left(\overrightarrow{C^{*}_jO},\overrightarrow{C^{*}_jP^{*}_{j+1}}\right)+\left(\overrightarrow{C^{*}_{j+1}P^{*}_{j+1}},\overrightarrow{C^{*}_{j+1}O}\right)=\delta\phi^*_j-\frac{\pi}{2}
\end{eqnarray*}
As $\epsilon_j=1$, it can be rewritten:
\begin{eqnarray*}
\left(\overrightarrow{C^{*}_jO},\overrightarrow{C^{*}_jP^{*}_{j+1}}\right)+\left(\overrightarrow{C^{*}_{j+1}P^{*}_{j+1}},\overrightarrow{C^{*}_{j+1}O}\right)=\delta\phi^*_j-\frac{\pi}{2}+\pi(1-\epsilon_j).
\end{eqnarray*}
\begin{eqnarray}
\label{formule_dphid}
\end{eqnarray}
\item If $\epsilon_j=-1$ then $\delta\phi^*_j<0$, in the quadrangle $C^{*}_jP^{*}_{j+1}C^{*}_{j+1}O$, $$\left(\overrightarrow{OC^{*}_{j+1}},\overrightarrow{OC^{*}_{j}}\right)=-\delta\phi^*_j$$ is positive, thus $\left(\overrightarrow{C^{*}_jO},\overrightarrow{C^{*}_jP^{*}_{j+1}}\right)$, $\left(\overrightarrow{C^{*}_{j+1}P^{*}_{j+1}},\overrightarrow{C^{*}_{j+1}O}\right)$ and $\left(\overrightarrow{P^{*}_{j+1}C^{*}_j},\overrightarrow{P^{*}_{j+1}C^{*}_{j+1}}\right)$ are also positive quantities. It means:
\begin{eqnarray}
\left(\overrightarrow{OC^{*}_{j+1}},\overrightarrow{OC^{*}_{j}}\right)_i&=&\left \lfloor{\left(\overrightarrow{OC^{*}_{j+1}},\overrightarrow{OC^{*}_{j}}\right)}\right \rfloor_{[0,2\pi]},\label{angoquadra1po}\\
\left(\overrightarrow{C^{*}_jO},\overrightarrow{C^{*}_jP^{*}_{j+1}}\right)_i&=&\left \lfloor{\left(\overrightarrow{C^{*}_jO},\overrightarrow{C^{*}_jP^{*}_{j+1}}\right)}\right \rfloor_{[0,2\pi]},\label{angoquadra2po}\\
\left(\overrightarrow{P^{*}_{j+1}C^{*}_j},\overrightarrow{P^{*}_{j+1}C^{*}_{j+1}}\right)_i&=&\left \lfloor{\left(\overrightarrow{P^{*}_{j+1}C^{*}_j},\overrightarrow{P^{*}_{j+1}C^{*}_{j+1}}\right)}\right \rfloor_{[0,2\pi]},\label{angoquadra4po}\\
\left(\overrightarrow{C^{*}_{j+1}P^{*}_{j+1}},\overrightarrow{C^{*}_{j+1}O}\right)_i&=&\left \lfloor{\left(\overrightarrow{C^{*}_{j+1}P^{*}_{j+1}},\overrightarrow{C^{*}_{j+1}O}\right)}\right \rfloor_{[0,2\pi]}.\label{angoquadra3po}
\end{eqnarray}
As $\delta\phi^*_j$ in $[-\pi,0]$, (\ref{angoquadra1po}) can be rewritten:
\begin{eqnarray}
\left(\overrightarrow{OC^{*}_{j+1}},\overrightarrow{OC^{*}_{j}}\right)_i&=&-\delta\phi^*_j.\label{angoquadra1bon2}
\end{eqnarray}
Substituting (\ref{angle_normal}) in (\ref{angoquadra4po}) gives:
\begin{eqnarray}
\left(\overrightarrow{P^{*}_{j+1}C^{*}_j},\overrightarrow{P^{*}_{j+1}C^{*}_{j+1}}\right)_i=\frac{\pi}{2}.\label{angoquadra4bon2}
\end{eqnarray}
Using (\ref{angoquadra1bon2}) and (\ref{angoquadra4bon2}) in (\ref{quadrangle}), it reads:
\begin{eqnarray*}
\left(\overrightarrow{C^{*}_jO},\overrightarrow{C^{*}_jP^{*}_{j+1}}\right)_i+\frac{\pi}{2}+\left(\overrightarrow{C^{*}_{j+1}P^{*}_{j+1}},\overrightarrow{C^{*}_{j+1}O}\right)_i-\delta\phi^*_j=2\pi
\end{eqnarray*}
which is equivalent to:
\begin{eqnarray*}
\left(\overrightarrow{C^{*}_jO},\overrightarrow{C^{*}_jP^{*}_{j+1}}\right)+\left(\overrightarrow{C^{*}_{j+1}P^{*}_{j+1}},\overrightarrow{C^{*}_{j+1}O}\right)=\delta\phi^*_j+\frac{3\pi}{2}.
\end{eqnarray*}
As $\epsilon_j=-1$, it can be rewritten:
\begin{eqnarray*}
\left(\overrightarrow{C^{*}_jO},\overrightarrow{C^{*}_jP^{*}_{j+1}}\right)+\left(\overrightarrow{C^{*}_{j+1}P^{*}_{j+1}},\overrightarrow{C^{*}_{j+1}O}\right)=\delta\phi^*_j-\frac{\pi}{2}+\pi(1-\epsilon_j).
\end{eqnarray*}
\begin{eqnarray}
\label{formule_dphidd}
\end{eqnarray}
\end{itemize}
The curvature along the arc $\mathcal{A}_{j}$ is constant equal to $\frac{o_j}{R_j}$ and (\ref{orth90}) tells the oriented angle between $\mathcal{A}_{j}$ and $\mathcal{A}_{j+1}$ is $\frac{\pi}{2}$. The Gauss-Bonnet formula (\cite{Cartan1967}) along the contour of the CCAOC thus reads:
\begin{eqnarray}
\gamma=\sum_{j\in\bZ/N\bZ}\left(\int_{\mathcal{A}_{j}}\frac{o_jds}{R_j}+\frac{\pi}{2}\right).\label{Gauss_bonnet}
\end{eqnarray}
Moreover the definition of a CCAOC specifies the orientation of $\mathcal{P}$ is counterclockwise: if $\rho$ is inferior to 1, for any $j$, $O$ is interior to $\mathcal{C}_{j+1}$ and the orientation of $\mathcal{P}$ implies  $o_j=1$ else if $\rho$ is superior to 1, for any $j$, $O$ is exterior to $\mathcal{C}_{j+1}$  and the orientation of $\mathcal{P}$ implies $o_j=-1$.
\begin{enumerate}
\item If $\rho\leq1$, the length of the CCL $\mathcal{C}_{j}$ is not in general bounded; following the arc $\mathcal{A}_{j}$ on the cone can imply turning several time around the circle $\mathcal{C}^*_{j}$ in the plane. There exists a positive integer (possibly null) $k_j$ (corresponding to the number of completed turns) such as:
\begin{eqnarray}
\int_{\mathcal{A}_{j}}\frac{o_jds}{R_j}=\left(\overrightarrow{C^{*}_jP^{*}_j},\overrightarrow{C^{*}_jP^{*}_{j+1}}\right)+2k_j\pi.\label{Gauss_bonnetrhosup}
\end{eqnarray}
We note:
\begin{eqnarray}
m=\sum_{j\in\bZ/N\bZ} k_j.\label{def_mrho_inf1}
\end{eqnarray}

\item If $\rho>1$, the length of the CCL $\mathcal{C}_{j}$ is bounded and equal to the length of $\mathcal{C}^*_{j}$;  the contour can turn at most one time around $\mathcal{C}^*_{j}$:
\begin{eqnarray}
\int_{\mathcal{A}_{j}}\frac{o_jds}{R_j}=\left(\overrightarrow{C^{*}_jP^{*}_j},\overrightarrow{C^{*}_jP^{*}_{j+1}}\right).\label{Gauss_bonnet0rhoinf}
\end{eqnarray}
We note:
\begin{eqnarray}
m=0.\label{def_mrho_sup1}
\end{eqnarray}
\end{enumerate}
The following sum can be decomposed into:
\begin{eqnarray*}
\sum_{j\in\bZ/N\bZ}\left(\overrightarrow{C^{*}_jP^{*}_j},\overrightarrow{C^{*}_jP^{*}_{j+1}}\right)=\sum_{j\in\bZ/N\bZ}\left(\left(\overrightarrow{C^{*}_jP^{*}_j},\overrightarrow{C^{*}_jO}\right)+\left(\overrightarrow{C^{*}_jO},\overrightarrow{C^{*}_jP^{*}_{j+1}}\right)\right)
\end{eqnarray*}
and rearranged into:
\begin{eqnarray*}
\sum_{j\in\bZ/N\bZ}\left(\overrightarrow{C^{*}_jP^{*}_j},\overrightarrow{C^{*}_jP^{*}_{j+1}}\right)=\sum_{j\in\bZ/N\bZ}\left(\left(\overrightarrow{C^{*}_{j+1}P^{*}_{j+1}},\overrightarrow{C^{*}_{j+1}O}\right)+\left(\overrightarrow{C^{*}_jO},\overrightarrow{C^{*}_jP^{*}_{j+1}}\right)\right).
\end{eqnarray*}
Substituting with (\ref{formule_dphid},\ref{formule_dphidd}) gives:
\begin{eqnarray}
\sum_{j\in\bZ/N\bZ}\left(\overrightarrow{C^{*}_jP^{*}_j},\overrightarrow{C^{*}_jP^{*}_{j+1}}\right)=\sum_{j\in\bZ/N\bZ}\left(\delta\phi^*_j-\frac{\pi}{2}+\pi(1-\epsilon_j)\right).
\end{eqnarray}
Finally substituting into (\ref{Gauss_bonnet}) gives:
\begin{eqnarray*}
\gamma=\sum_{j\in\bZ/N\bZ}\left(\delta\phi^*_j-\frac{\pi}{2}+\pi(1-\epsilon_j)\right)+N\frac{\pi}{2}+2m\pi
\end{eqnarray*}
which simplifies:
\begin{eqnarray*}
\gamma=\sum_{j\in\bZ/N\bZ}\left(\delta\phi^*_j+\pi(1-\epsilon_j)\right)+2m\pi.
\end{eqnarray*}
\end{enumerate}
\end{proof} 

\begin{thebibliography}{3}
\bibitem{Errera1886} L Errera. Sur une condition fondamentale d'equilibre des cellules vivantes. [On a
fundamental condition of equilibrium for living cells.]. C R Hebd Seances Acad Sci 103:822D824 (in French).
\bibitem{Besson2011}  Besson, S.,  Dumais, J. (2011). Universal rule for the symmetric division of plant cells. Proceedings of the National Academy of Sciences, 108(15), 6294-6299.
\bibitem{Louveaux2016} Louveaux, M., Julien, J. D., Mirabet, V., Boudaoud, A.,  Hamant, O. (2016). Cell division plane orientation based on tensile stress in Arabidopsis thaliana. Proceedings of the National Academy of Sciences, 201600677.
\bibitem{Wiener1914} N Wiener. The shortest line dividing an area in a given ratio. Proc Cam Philos Soc. 18:56D58
\bibitem{Wang2015} Y Wang, M Dou, Z Zhou. The fencing problem and Coleochaete cell division. J Math Biol. 4:893-912.
\bibitem{Couturier} E Couturier, A Lisee, P Llanos, S Besson, K Kamrin, J Dumais. In preparation.
\bibitem{Hofmeister1868} W. Hofmeister, Allgemeine Morphologic der Gewachse, Handbuch der Physiologishen Boranik (I Engelman, Leipzig, 1868), pp. 405-664. 
\bibitem{Iterson1907}G. Van lterson, Mathematische und Microscopisch Anatomische Studien uber Blattstellungen, nebst Betrschungen uber den Schalenbau der Miliolinen (Gustav-Fisher-Verlag, lena, 1907).
\bibitem{Cartan1967} H. Cartan, Cours de calcul différentiel (Hermann, 1967). (in French).
\end{thebibliography}
\end{document}